\documentclass[12pt, oneside, reqno]{article}
 \usepackage{graphicx}
\usepackage{amsmath,amsthm,amsfonts,amscd,amssymb,mathrsfs}
\usepackage{authblk}
\usepackage{enumerate}
\usepackage[all]{xypic}
\usepackage{algorithm}
\usepackage{algpseudocode}
\usepackage{booktabs} % for nicer tables
\usepackage{bbm}
\usepackage{physics}
\usepackage{tikz}
\usepackage{multirow}
\usepackage{nicefrac}
\usepackage{xfrac}
\usepackage{makecell}
\usepackage{geometry}
\usepackage{array}
\usepackage{tabularx}
\usepackage{mathdots}
\usepackage{mathtools}
\usepackage{hyperref}
\hypersetup{colorlinks=true,citecolor=blue, urlcolor= cyan}

\newcolumntype{C}{>{\centering\arraybackslash}m{6em}}

\newtheorem{theorem}{Theorem}[section]
\newtheorem{definition}[theorem]{Definition}

\newtheorem{proposition}[theorem]{Proposition}
\newtheorem{obs}[theorem]{Observation}

\theoremstyle{definition}

\usepackage[backend=biber,
style=numeric,
sorting=none
]{biblatex}

\title{A Quantum Algorithm For Computing Contextuality Bounds}

\author[$\dagger$]{Colm Kelleher}
\author[$\dagger,\ddagger$]{Frédéric Holweck}

\affil[$\dagger$]{Laboratoire Interdisciplinaire Carnot de Bourgogne, ICB/UTBM, UMR 6303 CNRS, Universit\'e de Technologie de Belfort-Montbéliard, 90010 Belfort Cedex, France}
\affil[$\ddagger $]{Department of Mathematics and Statistics, Auburn University, Auburn, AL, USA}

\date{}

\begin{document}
\maketitle
%\tableofcontents

\begin{abstract}
Quantum contextuality is a limitation on deterministic hidden variable models, testable in measurement scenarios where outcomes differ under quantum or classical descriptions due to a common set of constraints. When considering measurements of $N$-qubit spin operators, constraints arise from commutation relations and classical bounds are determined by the \textit{degree} of contextuality, an NP-hard quantity to compute in general, related to the larger class of optimisation problems known as \texttt{MaxLin2}. In this work we give a quantum algorithm based on Grover’s search algorithm, computing the degree of contextuality in $O(\sqrt{n} \log\log{n})$ in $n$ states, a speedup over classical brute force method. We also study variations of Grover which encode the relevant information in the phases of the basis states, reducing circuit width and depth requirements with indicative complexity of $O(n^{\frac{1}{3}}(\log{n})^{2}\log\log{n})$. Testing on contemporary quantum backends with the IBM quantum experience gives inconclusive outputs due to noise-induced errors, an issue hopefully fixed by future hardware.
\end{abstract}

\section{Introduction}
A set of linear constraints on a given input set over $\mathbbm{F}_{2}$ does not often offer a full solution. In fact just finding the maximal number of satisfiable constraints, let alone an input that delivers them, is known as \texttt{Max Lin 2}\cite{maxlin2} and is NP-Complete. In this article we focus on a quintessential family of such problems and offer a quantum algorithm that computes the input of maximal satisfaction, in sub-polynomial query time.

The problem we are interested in is computing the \textit{degree of contextuality}, a quantity associated to quantum systems that will be explained in \S \ref{sec:contextuality}, with it's relation to \texttt{MaxLin2} in \S \ref{sec:max_lin_2} and examples used later in the text introduced in \S \ref{sec:geometries}. We describe the algorithm itself in Section \ref{sec:algorithm}. We then give results of running this algorithm on the IBM Quantum Experience\cite{ibmq} for both simulations and real quantum backends in Section \ref{sec:results_chpt_6}, and then explore variations in Section \ref{sec:quasi_grover}.

\subsection{Quantum Contextuality}\label{sec:contextuality}
Quantum contextuality is a no-go limitation on classical hidden variable models, first described by Kochen and Specker\cite{Kochen_specker} and Bell\cite{bell_problem_1966}. In its original form it proved the impossibility of assigning a measurement-independent configuration - described by deterministic hidden variables - to any quantum state living in a Hilbert space of dimension $n \geq 3$. In later years it was reformulated around measurement operators and state-independence, most famously by Peres\cite{Peres90} and Mermin\cite{Mermin93} who gave a geometric proof of the non-existence of non-contextual hidden variables (NCHV) in dimension $n \geq 4$. NCHV models are those using deterministic hidden variables where the pre-determined measurement outcome is not dependent on the choice of context the measured observable is considered part of.

Recent work has emphasised the geometric nature of contextuality (see for example \cite{cabello_proposed_2010, Budroni21, Hexagon_3_qubits, muller_multi_qubit_2022, henri_contextuality_2022}) and its link with nonlocality \cite{cabello_converting_2021}, quantum computing \cite{howard_contextuality_2014}, random access codes \cite{gatti_random_2023}, and recently quantum error-correcting codes \cite{khu2025contextuality}. Contextual inequalities can be formed similar to nonlocality ones, and tested on photonic devices \cite{xu_experimental_2022} and noisy intermediate-scale quantum (NISQ) computers \cite{Kelleher_2_qubit_games, kelleher_exploiting_2024}. Such demonstrations illustrate empirically the lack of NCHV models, but they require known classical bounds to violate. Such bounds rely on the degree of contextuality of the system in question, defined in Def. \ref{def:degree}, which is NP-hard to compute in the general case. In this section we construct the constraints needed to show contextuality, and then give examples of specific geometric representations that will be used to test the algorithm described in the following sections.

Consider the set of Pauli 1-qubit spin operators including identity,
\begin{equation}
    \mathcal{P} := \{X, Y, Z, I\}
\end{equation}
\begin{equation}
    X = \begin{pmatrix}
        0 & 1 \\
        1 & 0
    \end{pmatrix} \quad Y = \begin{pmatrix}
        0 & -i \\
        i & 0
    \end{pmatrix} \quad Z = \begin{pmatrix}
        1 & 0 \\
        0 & -1
    \end{pmatrix} \quad 
\end{equation}
and the set of all tensor products of $N$ Pauli operators, excluding the trivial one $\mathbbm{1} := I^{\otimes N}$:
\begin{equation}
    \mathcal{O} := \{ \sigma_{1}\otimes \dots \otimes \sigma_{N},\quad \sigma_{i} \in \mathcal{P} \} \backslash I^{\otimes N}
\end{equation}
This set, and certain subsets thereof, contain triples of operators that mutually commute and whose product is the identity up to a sign. We refer to each such triple as a \textit{context} (not to be confused with the broader term in the literature, which generally supposes a maximal set of mutually-commuting operators with product identity). We can label each context by its constituent operators
\begin{equation}\label{eq:context_definition}
    \begin{split}
    \mathcal{C}_{i} &= \{\mathcal{O}_{i_{1}}, \mathcal{O}_{i_{2}}, \mathcal{O}_{i_{3}}\} \\
     \quad \mathcal{O}_{i_{1}}\mathcal{O}_{i_{2}}\mathcal{O}_{i_{3}} = \pm \mathbbm{1}, \quad \mathcal{O}_{i_{j}} &\in \mathcal{O}, \quad [\mathcal{O}_{i_{j}}, \mathcal{O}_{i_{k}}] = 0 \;\;\forall j, k \in \{1,2,3\}
\end{split}
\end{equation}
The commutation of the operators ensures that the defining functional relationship in \eqref{eq:context_definition} is also obeyed by their eigenvalues:
\begin{equation}\label{eq:eigenvalue_relation}
    e(\mathcal{O}_{i_{1}})e(\mathcal{O}_{i_{2}})e(\mathcal{O}_{i_{3}}) = \pm 1
\end{equation}
where $e(\mathcal{O}_{i_{j}}) = \pm 1$ is also interpreted as the outcome of measuring a given state in the $\mathcal{O}_{i_{j}}$ spin direction. Equation \eqref{eq:eigenvalue_relation} then gives a constraint on measurement outcomes for compatible measurements, whose product is known by the context product. We will intuitively refer to contexts whose operators multiply to $+\mathbbm{1}$ as \textit{positive} and those to $-\mathbbm{1}$ as \textit{negative}.

\subsection{Product Constraints \& \texttt{MaxLin2}}\label{sec:max_lin_2}
As operators can appear in multiple contexts, their measurement outcomes are subject to multiple constraints at once. The notion of contextuality is seen by the fact that not all constraints can be simultaneously satisfied by classical values. This is most readily seen in the celebrated Peres-Mermin Magic Square\cite{Peres90}, a collection of 9 operators sharing 6 contexts, arranged in a grid formation - see for example Fig. \ref{fig:grid_example}.

\begin{figure}[h!]
    \centering
    \includegraphics[width=0.3\textwidth]{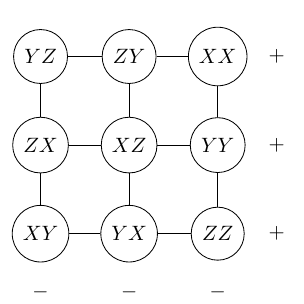}
    \caption{An example of a Peres-Mermin Magic Square for $N=2$. 2-qubit operators written with tensor product suppressed. Contexts are given by lines of 3 operators, positive ones labeled by ``$+$" (horizontal) and negative by ``$-$" (vertical). Any classical HV assignment of $\pm1$ to the operators will inevitably lead to at least one context constraint unsatisfied.}
    \label{fig:grid_example}
\end{figure}

The product constraint along a context translates to classical constraints on the measurement outcomes. Writing these constraints for the square:

\begin{equation}\label{eq:mermin_functional_relationships}
\begin{split}
	e(YZ)\cdot e(ZY)\cdot e(XX) &= e(II) = 1 \\
	e(ZX)\cdot e(XZ)\cdot e(YY) &= e(II) = 1 \\
	e(XY)\cdot e(YX)\cdot e(ZZ) &= e(II) = 1 \\
	e(YZ)\cdot e(ZX)\cdot e(XY) &= e(-II) = -1 \\
	e(ZY)\cdot e(XZ)\cdot e(YX) &= e(-II) = -1 \\
	e(XX)\cdot e(YY)\cdot e(ZZ) &= e(-II) = -1
\end{split}
\end{equation}

There is no consistent assignment to the outcomes $e(\mathcal{O}_{i_{j}})\in \{ \pm 1\}$ that can satisfy more than 5 of the 6 above constraints. Such a pre-measurement classical assignment will be referred to as a \textit{hidden variable} (HV) assignment, as it can be considered as arising deterministically from some variables unseen by quantum mechanics.

One can readily translate each contextual constraint to a set of constraints over $\mathbbm{F}_{2}$, where for the context $\mathcal{C}_{i}$:

\begin{equation}
\begin{split}
e(\mathcal{O}_{i_{j}}) &= (-1)^{a_{i_{j}}} \\
\Pi_{j}e(\mathcal{O}_{i_{j}}) = (-1)^{b_{i}} &\mapsto \sum_{j}a_{i_{j}} = b_{i} \text{ mod }2
\end{split}
\end{equation}

A collection of quantum operators may satisfy all contextual product relations given by \eqref{eq:context_definition}, yet no such HV assignment exists that satisfies all associated classical relations,
\begin{equation}
\Pi_{j}e(\mathcal{O}_{i_{j}}) = (-1)^{b_{i}}, \quad \forall \; i
\end{equation}

In such a scenario we say that the set of operators and context constraints, or geometry representing them as defined below, is \textit{contextual}.

\subsection{Geometric Description}\label{sec:geometries}
The Peres-Mermin Magic Square is a geometric representation of the \texttt{MaxLin2} problem set \eqref{eq:mermin_functional_relationships}, with all constraints satisfied using quantum operators and at most 5 out of 6 satisfied for any classical assignment. In general any collection of $N$-qubit operators can be formed into a geometry, with vertices labeled by operators and lines formed from contexts. Each line will have an associated sign, informing its contextual constraint, all of which are satisfied with quantum operators. We can then give a formal definition of the degree of contextuality:
\begin{definition}\label{def:degree}
    The \textit{degree of contextuality} $d$ for an $N$-qubit geometry is the minimal number of line constraints left unsatisfied by any classical HV vertex assignment. 
\end{definition}

The degree of contextuality is a measure of how non-classical a set of operators are, as all context constraints are satisfied by quantum operators but not by classical HVs.

As the square Fig. \ref{fig:grid_example} contains 6 line constraints out of which only 5 are simultaneously satisfiable by classical HVs, it has degree 1. For another example, one can take the set of all nontrivial 2-qubit spin operators and form the set of all contexts associated with them, which is known as the doily. It contains 15 points (operators), 15 lines (contexts), out of which 3 are negative, and each point lies on 3 lines. See Fig. \ref{fig:doily}, left. Removing five non-intersecting lines gives the \textit{two-spread}, a 15-point, 10-line geometry with degree of contextuality $d=1$ (see Fig. \ref{fig:doily}, right).

As a final geometry of interest to introduce here, one can take the set of all commutation relations between 3-qubit operators of the form $\mathcal{O} = I\otimes \sigma_{i} \otimes \sigma_{j}$, with $\sigma_{i} \in \{X, Y, Z\}$ and allow permutations of qubits. This gives a 27-point, 45-line, $d=9$ geometry known as the \textit{eloily}, pictured in Fig. \ref{fig:eloily}. The eloily has been studied in relation to quantum games \cite{kelleher_exploiting_2024}, as it has the interesting property that the $d=9$ invalid lines furnishing an optimal solution are pairwise disjoint and intersect each point exactly once. The grid, doily and eloily are each \textit{generalised quadrangles} in that none contain triangles, and they are the only such quadrangles with 3-point lines. 

\begin{figure}[h!]
    \centering
    \includegraphics[width=.45\textwidth]{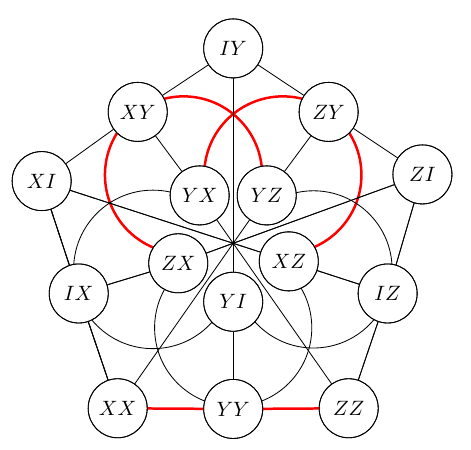}
    \includegraphics[width=.45\textwidth]{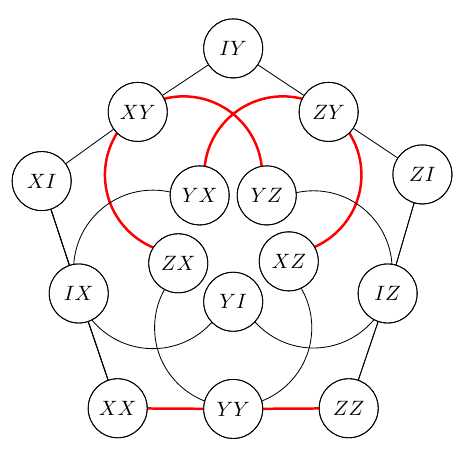}
    \caption{Left: The doily geometry of all 15 nontrivial 2-qubit operators. Negative lines are indicated in red. The doily has degree 3. Right: The two-spread geometry, with 5 disjoint lines from the doily removed - in this case the central lines from outermost edge vertex to outermost edge midpoint. The number of negative lines varies but the degree is always 1.}
    \label{fig:doily}
\end{figure}

\begin{figure}[h!]
    \centering
    \includegraphics[width=.45\textwidth]{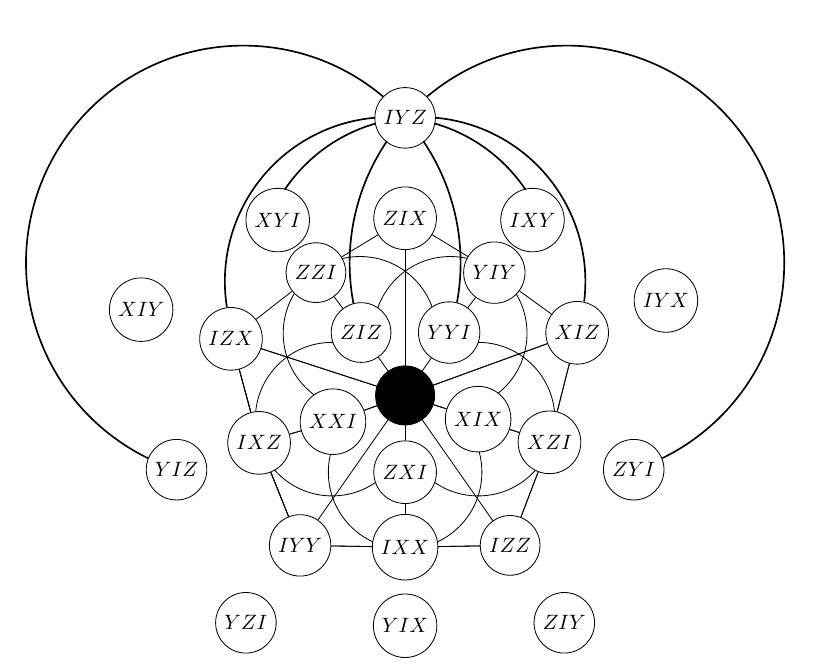}
    \includegraphics[width=.45\textwidth]{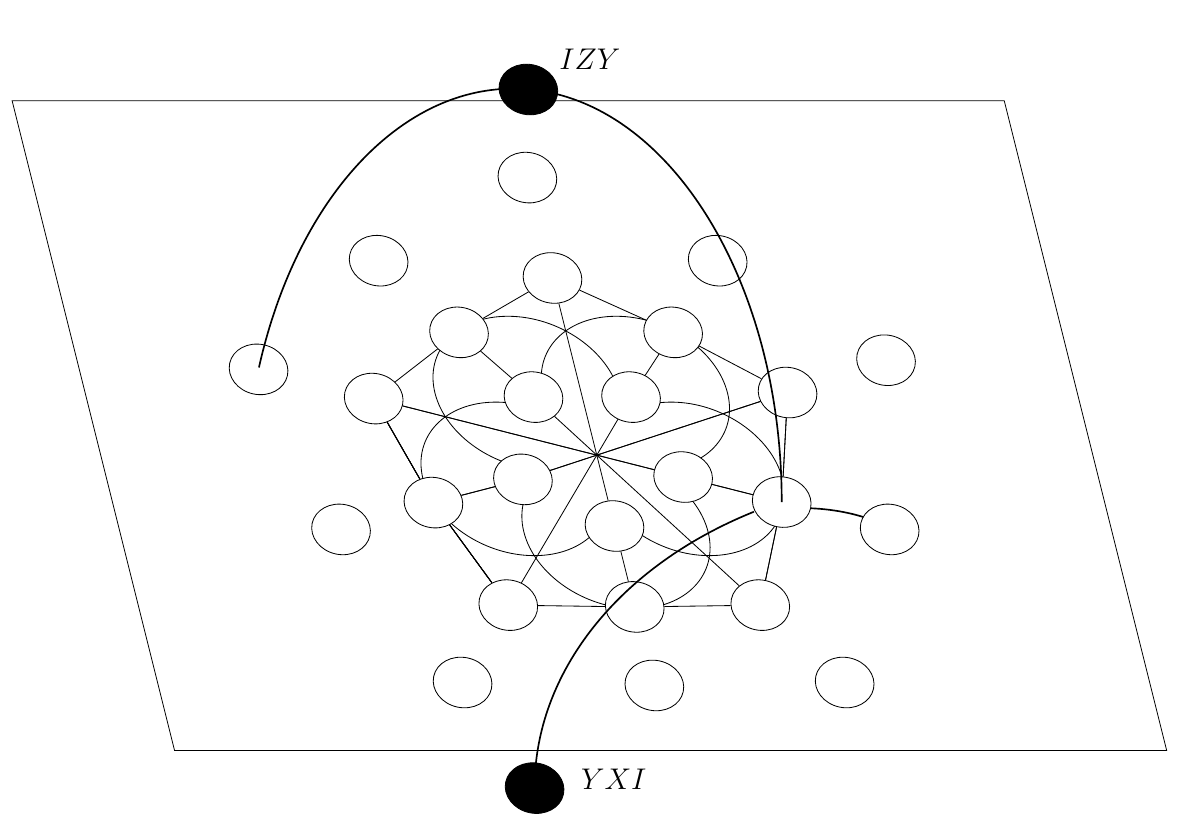}
    \caption{The eloily, with representative lines shown and labelled by 3-qubit operators - for full set of lines rotate by $\frac{2\pi}{5}$. Right shows selected lines between black points not co-planar with white. There are 9 negative lines, not indicated. It can be seen that a copy of the doily (labelled with 3-qubit operators, also containing 3-qubit grids) are embedded within the eloily.}
    \label{fig:eloily}
\end{figure}

In considering general $N$-qubit operators, the full set of commutation relations furnishes a geometry known as the \textit{symplectic polar space} of rank $2N-1$ over the field of characteristic 2, denoted $\mathcal{W}(2N-1,2)$. Of interest to us in this general case is the number of points, line and lines per point, where we consider (as in the exampled geometries above) lines to be 3 mutually-commuting operators with product $\pm \mathbbm{1}$. This general case has the structure:

\begin{equation}\label{eq:geometry_points_lines}
\mathcal{W}(2N-1,2) \text{ contains } \begin{cases}
	4^{N} - 1 \text{ points} \\
	(4^{N} - 1)(4^{N-1} - 1)/3 \text{ lines} \\
	4^{N-1} - 1 \text{ lines per point} \\
\end{cases}
\end{equation}

See \cite{muller_new_2024} for recent computations of $d$ for up to 7-qubit symplectic polar spaces and subgeometries.

What is of note is that the quantum contextuality of a geometry does not depend on the particular quantum labelling of its vertices, only by the point-line incidence structure of the geometry and distribution of negative lines. This is because the degree is a bound on classical HVs, which are subject these constraints.

\subsection{Applications of $d$}
When demonstrating contextuality in practice, experimental setups are designed such that any classical NCHV description presents an upper bound that can be violated by quantum mechanics. Such bounds rely on the degree, for example in the Rio Negro inequality\cite{cabello_proposed_2010}, tested in \cite{muller2025hexagons}, has for a geometry with $L$ lines a classical bound of 

\begin{equation}
\chi_{NCHV} \leq L - 2d
\end{equation}

Likewise there exist pseudotelepathy games that demonstrate contextuality, the most famous being the Mermin game \cite{cabello2001bell}, generalised via geometries described above in \cite{Kelleher_2_qubit_games, kelleher_exploiting_2024}. For a geometry with $V$ vertices and $l_{p}$ lines per point, a $P$-player game can be constructed with classical upper bounds:

\begin{equation}
\begin{split}
\omega_{ll} &\leq_{NCHV} 1 - \frac{{\ell_{p} - 1 \choose 1}}{{\ell_{p} \choose P}}\cdot \frac{d}{V} \\
\omega_{pl} &\leq_{NCHV} 1 - \frac{d}{L}\cdot\frac{1}{3} 
\end{split}
\end{equation}

Here the $ll$ and $pl$ bounds correspond to games where each player receives a line from the geometry, and one a line the other a point, respectively. It is clear from the above that the degree $d$ is a necessary quantity in a wide range if not all contextuality tests. All of the above examples have shown violations of these classical bounds on modern day NISQ computers, using geometries as shown in this work, in \cite{kelleher_empirical_2025}.

\section{The Algorithm}\label{sec:algorithm}
When computing the degree one could na{\"i}vely give all possible HV vertex assignments and find the minimal number of unsatisfied constraints. This would of course be unfeasible for large $N$-qubit geometries, which have up to $4^{N} -1$ vertices.

In 1999 D{\"u}rr and H{\o}yer\cite{durr_quantum_1999_grover} gave a fast quantum algorithm to identify the minimal element of an unsorted list, based on Grover's original quantum search algorithm\cite{grover_fast_1996}, in $O(\sqrt{n}\log{n})$, using the measured output at each iteration in a binary search. We use this approach combined with our data from our chosen geometry to compute $d$.

The idea is straightforward. We take as input a given geometry $\mathcal{G}$, with $V$ vertices and $L$ line constraints (``negative" lines known). Assigning to each vertex a qubit encoding its HV value, one can place the overall state into a superposition of all HV assignments over $\mathcal{G}$. From this we exploit Grover's algorithm, marking all states with invalid line counts $x$ below some chosen threshold $y$, then measuring such a state and updating our threshold $y \mapsto x$ before repeating. This follows \cite{durr_quantum_1999_grover}, the details we need to cover are the computation of invalid lines $x$, and determining which states to mark via the check $x \leq y$. Each iteration requires $O(\sqrt{n})$ oracle checks, and for binary search we need $O(\log{L}) = O(\log\log{n})$ iterations, from $L \sim V^{2}$ in $\mathcal{W}(2N-1,2)$, see equation \eqref{eq:geometry_points_lines}.

%We set up a register with $V$ qubits, encoding a superposition of all HV assignments on the vertices. From this register and the parameters of the geometry we can compute the number of invalid constraints, $x$, for any given vertex state and encode that in a separate register, entangled with the first. Then selecting some threshold $y$ we can implement an oracle marking all states satisfying $x \leq y$. Running Grover's algorithm with this oracle, one measures a marked input state with some $y^{(1)} \leq y$ invalid lines, which we use to update our threshold guess and run again. Eventually we find a stable output $y^{(n)} = \text{min}(x)$.

In Section \ref{sec:geometric_reg} we describe the segment of the circuit that computes $x$ for a given assignment, via operator $U_{\mathcal{G}}$. In Section \ref{sec:comparison_reg} we describe the segment that marks states with $x \leq y$, via $U_{\mathcal{C}}$. Our oracle is the composition of these, $U = U_{\mathcal{G}}^{\dagger}U_{\mathcal{C}}'U_{\mathcal{C}}U_{\mathcal{G}}$, with $U_{\mathcal{C}}'$ explained later.

\subsection{The Geometric Segment}\label{sec:geometric_reg}
First we explain the segment of the circuit $U_{\mathcal{G}}$ that describes the geometry $\mathcal{G}$. We need three qubit registers: one containing a qubit for each vertex (called $V$); one containing a qubit for each line ($L$), and one counting the number of invalid lines ($X$). The third register will need $\text{log}_{2}L$ qubits. The letters used to denote these registers will also denote their qubit sizes. We initialise all qubits in these registers to the state $\ket{0}$. See Fig. \ref{fig:circuit_V_L_x}, and below text for explanation. 

For the $V$ register, we place each vertex into a superposition using a Hadamard gate $H$. For any set of vertices $\{V_{i_{1}}, V_{i_{2}}, V_{i_{3}}\}$ sharing a line $L_{i}$ in our geometry, we place a ``gate" between the associated $V$ qubits and the $i$th $L$ qubit. This gate will flip the state of $L_{i}$ qubit if the states of the $\{V_{i_{j}}\}$ qubits render their line invalid. Call this gates $P$ (for positive $L_{i}$) and $N$ (for negative) gates. The are defined in Fig. \ref{fig:circuit_P_N_def}.
\begin{figure}[h!]
     \centering
     \includegraphics[width=0.4\textwidth]{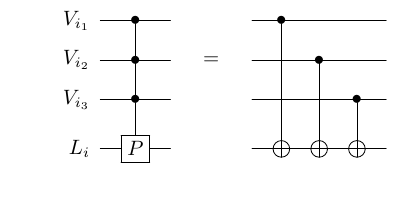}
     \includegraphics[width=0.45\textwidth]{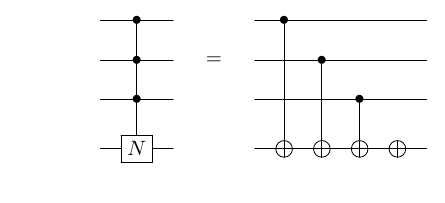}
     \caption{Definitions of the $P$ and $N$ gates for line $L_{i}$. The control qubits are in the $V$ register associated to points $V_{i_{j}}$ in the line $L_{i}$, which is associated with the target qubit in the $L$ register. They ensure the state of the target (initialised to $\ket{0}$) is in state $\ket{1}$ iff the line constraint is unsatisfied.}
     \label{fig:circuit_P_N_def}
 \end{figure}
Now the $L$ register will be in a superposition of all possible configurations of valid and invalid lines, entangled with the vertex assignments. 

To finish $U_{\mathcal{G}}$, we describe the $X$ register which counts the number of invalid lines. The state of the $\text{log}_{2}L$ qubits acts as a binary encoding of this number. First on the $2^{0}$ count wire in $X$, we apply $L$ CNOT targets, controlled by each of the qubits in $L$. Next on the $2^{1}$ wire, we apply ${L \choose 2}$ toffoli gates controlled by each pair of qubits in $L$. For the $2^{k}$ wire, we apply ${L \choose 2^{k}}$ targets controlled by each set of $2^{k}$ qubits in $L$. This leaves the $X$ register encoding in binary the count of $\ket{1}$ states in $L$, as desired. We require $\sum_{k=0}^{X}{L \choose 2^{k}}$ multi-control NOT gates for this. The implementation of $U_{\mathcal{G}}$ is expressed in Algorithm \ref{alg:U_G}. 

 \begin{figure}[h!]
     \centering
     \includegraphics[width=0.6\textwidth]{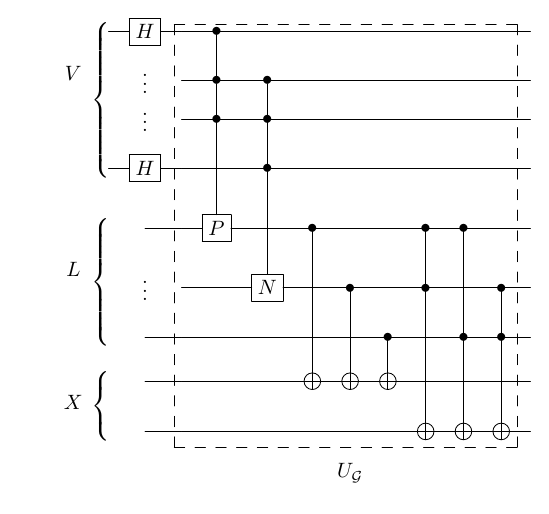}
     \caption{Circuit representing the $U_{\mathcal{G}}$ operator on the $V$, $L$ and $X$ registers, which counts the number of invalid lines for a given geometry $\mathcal{G}$ with superposition of HV assignments.}
     \label{fig:circuit_V_L_x}
 \end{figure}

The $\texttt{P}(\{q_{c}\}, q_{t}), \texttt{N}(\{q_{c}\}, q_{t}), \texttt{mcx}(\{q_{c}\}, q_{t})$ gates therein take as inputs $\{q_{c}\}$ control qubits and $q_{t}$ as the target, these being the $P$ gate, $N$ gate (Fig. \ref{fig:circuit_P_N_def}) and multi-control NOT gate. 

Note for clarity in the next section, we start indexing the $X$ qubits from the highest binary digit. The count of invalid lines is given by $x = x_{0}x_{1}\dots x_{l}$ with $l := \log_{2}{L}-1$, and $x_{0}$ being the binary digit representing $2^{l}$ invalid lines, $x_{1}$ encoding $2^{l-1}$ invalid lines, etc. We choose this encoding as it renders the binary digits in $x$ in the correct order, and leaves the for loops below a bit cleaner.
\begin{algorithm}[h!]
\caption{Implementation of $U_{\mathcal{G}}$}\label{alg:U_G}
\begin{algorithmic}
\For{$i \in (1, \dots, L)$}
    \If{line $L_{i}$ positive}
        \State $\texttt{circuit.P}(\{V_{i_{1}}, V_{i_{2}}, V_{i_{3}}\}, L_{i})$\;
    \Else
        \State $\texttt{circuit.N}(\{V_{i_{1}}, V_{i_{2}}, V_{i_{3}}\}, L_{i})$\;
      \EndIf
      \EndFor
\For{$k \in (0, \dots, l)$}
    \For{$\{L\}_{k} \subseteq L \text{ with } |\{L\}_{k}| = 2^{k}$} \Comment{All sets of $2^{k}$ lines in $L$}
        \State $\texttt{circuit.mcx}(\{L\}_{k}, x_{l-k})$\; \Comment{Flip the $k^{\text{th}}$ bit encoded in $X$ if all $\{L\}_{k}$ invalid}
        \EndFor
    \EndFor
\end{algorithmic}
\end{algorithm}

\subsection{The Comparison Segment}\label{sec:comparison_reg}
We now wish to implement $U_{\mathcal{C}}$, marking all states $\ket{j_{y}}$ in $V$ that leave $x \leq y$ for a chosen threshold $y$. We mark these input states by inducing a phase change $\ket{j_{y}} \rightarrow -\ket{j_{y}}$, allowing us to run Grover's algorithm to amplify the amplitudes of these states.

There exist already quantum circuits computing the truth value of comparisons such as $x \leq y$, for example in \cite{quantum_bitstring_comparison}, all requiring ancilla qubits. Here we describe a new comparison circuit that requires no additional qubits other than the ones encoding $x$ and $y$, and one encoding $\texttt{bool}(x\leq y)$.

This segment introduces a new register $Y$ of the same width $l+1$ as $X$, initialised to the binary value of $y$ via quantum NOT gates. We also introduce the qubit encoding $\texttt{bool}(x \leq y)$, initialised to $\ket{-}$ and labelled as such. We run the following Algorithm \ref{alg:U_C} which executes $U_{\mathcal{C}}$:
\begin{algorithm}[h!]
\caption{Implementation of $U_{\mathcal{C}}$}\label{alg:U_C}
\begin{algorithmic}
\State $\texttt{circuit.x}(\ket{-})$\;
\For{$i \in (0, \dots ,l)$}
        \State $\texttt{circuit.mcx\_var}(\{y_{0}, \dots , y_{i}\}, \{x_{i}\}, \ket{-})$\;
        \State $\texttt{circuit.cx}(x_{i}, y_{i})$
      \EndFor
\end{algorithmic}
\end{algorithm}
\\
The $\texttt{mcx\_var}(\{\overline{q_{c}}\}, \{q_{c}\}, q_{t})$ function takes as inputs $\{\overline{q_{c}}\}$ as flipped control qubits (i.e. activate iff they are all in the state $\ket{0}$), $\{q_{c}\}$ as regular control qubits, and $q_{t}$ as the target. The overall effect of Algorithm \ref{alg:U_C} is that it checks whether the binary number $x = x_{0}\dots x_{l}$ is less than or equal to the binary number $y = y_{0}\dots y_{l}$, and applies a NOT gate to the target $\ket{-}$ if so, inducing the desired phase change. It is expressed as a quantum circuit in Fig. \ref{fig:algo_y_geq_c_circuit}.
\begin{figure}[h!]
    \centering
    \includegraphics[width=0.6\textwidth]{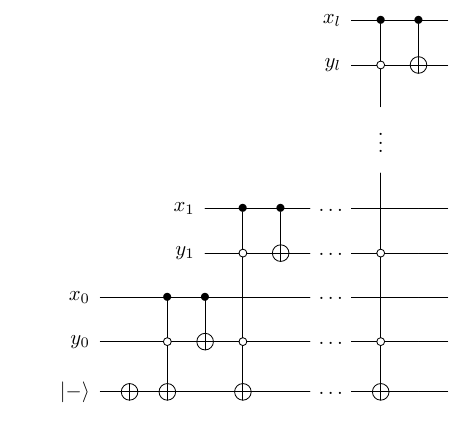}
    \caption{Quantum circuit showing $U_{\mathcal{C}}$ to check whether $x \leq y$, for $x = x_{0}x_{1}\dots x_{l}$, $y = y_{0}y_{1}\dots y_{l}$. If so, it flips the last wire. $X$ and $Y$ registers are intermingled for clarity.}
    \label{fig:algo_y_geq_c_circuit}
\end{figure}
In order to run the oracle multiple times we must restore the $Y$ register to its original value via $U_{\mathcal{C}}'$, see Algorithm \ref{alg:U_C_prime}.
\begin{algorithm}[h!]
\caption{Implementation of $U_{\mathcal{C}}'$, restoring $Y$ register to $y$}\label{alg:U_C_prime}
\begin{algorithmic}
\For{$i \in (0, \dots ,l)$}
        \State $\texttt{circuit.cx}(x_{l-i}, y_{l-i})$
      \EndFor
\end{algorithmic}
\end{algorithm}

For completion Algorithm \ref{alg:diffusion} provides the implementation of the Diffusion operator D.
\begin{algorithm}[h!]
\caption{Implementation of $\mathcal{D}$}\label{alg:diffusion}
\begin{algorithmic}
\State $\texttt{circuit.h}(V)$\;
\State $\texttt{circuit.x}(V)$\;
\State $\texttt{circuit.mcz}((V_{1}, \dots, V_{V-1}), V_{V})$\Comment{apply CZ gate with target as last qubit in $V$, controlled on all other $V$ qubits}\;
\State $\texttt{circuit.x}(V)$\;
\State $\texttt{circuit.h}(V)$\;
\end{algorithmic}
\end{algorithm}

\subsection{Analysis}\label{sec:analysis}
The result of the above is an oracle $U = U_{\mathcal{G}}^{\dagger}U_{\mathcal{C}}'U_{\mathcal{C}}U_{\mathcal{G}}$ that marks HV assignments giving an invalid count $x \leq y$. As per Grover, this is then followed by a diffusion operator $\mathcal{D}$ - this whole process repeated $t_{G} = \lfloor \frac{\pi}{4}\sqrt{\frac{n}{m}}\rfloor$ times for $n = 2^{V}$ total assignments and $m$ of them marked. Upon measurement on $V$ we are given a HV assignment with some $y^{(1)} \leq y$ invalid lines, which we use as a new threshold and run again. Taking each measurement outcome as a bisection on the sample space of possible solutions, one can consider the number of iterations needed as in $O(\log{L}) = O(\log{\log{n}})$. We refer to the number of measurements as $t_{A}$.

A brief note that for $m$ marked states and $n-m$ unmarked, Grover's algorithm amplifies the amplitudes of the smaller of the two sets. Not knowing a priori what the degree is, it is possible that we have marked a majority of states and return instead some new threshold $y^{(1)}$ with $y^{(1)} > \frac{L}{2}$. This is not an issue, as for 3 vertices per line we can simply flip all HV signs in such an assignment to find $L - y^{(1)}$ invalid lines, with $L-y^{(1)} \leq \frac{L}{2}$. If this value is below our initial guess $y$ we update the threshold and run again. The entire algorithm can then be written as follows in Algorithm \ref{alg:whole_alg}, where we take as an initial threshold the number of negative lines.
\begin{algorithm}[h!]
\caption{A quantum circuit to compute the degree $d$ of a geometry}\label{alg:whole_alg}
\begin{algorithmic}
\State $\texttt{circuit} \gets\texttt{QuantumCircuit}(V+L+X+Y+1, V)$\Comment{Circuit with qubits, classical bit sizes}; 
\State $y \gets |L_{\text{neg}}|$ 
\For{$t_{a} \in 1,\dots, t_{A}$}
    \State \texttt{circuit.h}$(V)$\;
    \State \texttt{circuit.assign}$(y, Y)$\;
    \For{$t_{g} \in 1,\dots,t_{G}$}
        \State \texttt{circuit.$U_{\mathcal{G}}$}$(V,L,X)$\;
        \State \texttt{circuit.$U_{\mathcal{C}}$}$(X,Y, \ket{-})$\;
        \State \texttt{circuit.$U_{\mathcal{C}}'$}$(X,Y)$\;
        \State \texttt{circuit.$U_{\mathcal{G}}^{\dagger}$}$(V,L,X)$\;
        \State \texttt{circuit.$\mathcal{D}$}$(V)$\;
    \EndFor;
    \State $S_{V} \gets $\texttt{circuit.measure}$(V)$\Comment{save input state with $x \leq y$ invalid lines}\;
    \State $x \gets $\texttt{invalid}$(S_{V})$\Comment{classically compute invalid lines}\;
    \State $y \gets \text{min}(x, L-x)$\;
\EndFor;
\State \Return $y$
\end{algorithmic}
\end{algorithm}

For maximal effectiveness the quantity $\frac{m}{n}$ is needed. This is in general not known, but can be estimated by assuming a normal distribution of invalid lines across assignments. As mentioned, any assignment with $x$ invalid lines can be paired with a corresponding assignment with $L-x$ lines, via flipping each HV. Thus any distribution will be symmetric around $\frac{L}{2}$. While the distribution will not be rigorously normal, as any contextual geometry will not have $x < d$ invalid lines for $d$ the degree, one can assume that the number of assignments of $x$ invalid lines increases with $x$ in $d \leq x \leq \frac{L}{2}$. This gives the estimation of $\frac{m}{n}$ as 
\begin{equation}
	\frac{m}{n} \approx \frac{1}{\sqrt{2\pi}}\int_{-\infty}^{z}e^{-t^{2}/2}dt
\end{equation}
with variance $\sigma$ suitably chosen for $P(y = 0) \approx 0$ in
\begin{equation}
	z = \frac{y - L/2}{\sigma}
\end{equation}

One can initialise with $y = |L_{\text{neg}}|$, and take that the algorithm will stabilise to $y = d$ with $t_{A} \in O(\log{L})$. This gives an overall time complexity of $O(\sqrt{n}\log{\log{n}})$.

\subsection{The Triangle Example}\label{sec:triangle}

Our examples are not constrained by geometries labeled by quantum operators, as this algorithm is suitable for a broader set of problems. Consider the geometry in Fig. \ref{fig:triangle_geom}, which consists of 3 vertices and 3 lines, the maximal set of satisfiable constraints being 2, and is not point-wise labellable by quantum operators satisfying all constraints.

\begin{figure}[h!]
    \centering
    \includegraphics[width=0.3\textwidth]{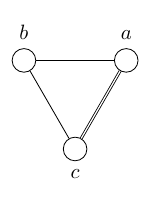}
    \caption{Geometry of 3 points, 3 lines, one negative in double font. Only 2 of the 3 line constraints can be satisfied by any HV assignment $(\pm 1)$ on the vertices.}
    \label{fig:triangle_geom}
\end{figure}

Here one can initialise with a threshold $y = 2$, and running the code on a simulator reveals an output (see Table \ref{tab:clean_sim_results}) of $y^{(1)} = 3$, indicative of a minority of unmarked elements being selected. This is because only $y^{(1)}=3$ has fewer HV assignments than $y^{(1)}\leq 2$, and Grover amplifies the smaller of the marked and unmarked sets of states.

\section{Results}\label{sec:results_chpt_6}
Algorithm \ref{alg:whole_alg} requires $V+L+X+Y+1$ qubits, restricting the geometries that can be tested on modern day quantum machines. In addition, without circuit transpilation there are around $G = 2(L+\sum_{k=0}^{l}{L \choose 2^{k}})+3X$ multi-control NOT gates per oracle implementation, each of which may be decomposed into many CNOTs depending on the ancilla method chosen. With transpilation for backend optimisation this can translate to many thousands of CNOT gates for even one iteration of $U$. Because of these width and depth requirements we are limited to testing only small geometries such as the grid (Fig. \ref{fig:grid_example}) on real backends and triangle (Fig. \ref{fig:triangle_geom}) on backends and transpiled simulators. Even then, transpilation-induced depth can easily increase circuit noise to levels surmounting any signal given by the output.

We first show in this section results on a clean, non-transpiled circuit to demonstrate the effectiveness of the algorithm. We follow this with results on both noisy, transpiled circuits and on real quantum backends for comparison. In the latter case, transpilation-induced noise eliminates any signal given by the algorithm. For both types of tests we use the IBM Quantum Experience Platform\cite{ibmq}, noting that the automated transpilation functions are semi-stochastic. 

\subsection{Results on a Clean Simulator}
The algorithm was run on both the triangle and grid geometries, using IBM's native \texttt{AerSimulator}. This was done with no noise model or transpilation as a test of the algorithm on an idealised, fully CNOT-connected quantum backend. For both geometries the experiment was run with 2,048 shots, with final HV assignment measured directly from the $V$ register and the number of invalid lines computed classically. The results are in Table \ref{tab:clean_sim_results}, with ``control" runs with no Grover segments ($t_{G}=0$) alongside algorithm results for comparison. In both geometries, the Grover segment clearly shifts likely measured assignments toward preferred $y^{(1)}$ states (where $y^{(1)} < y$ for the grid, and $y^{(1)} = 3$ for the triangle as explained in \S \ref{sec:triangle}).

\begin{table}[h!]
\centering
\begin{tabular}{|c|cccc|ccc|}
    \hline
    \multirow{2}{*}{Geometry} & \multirow{2}{*}{Backend} & 
    \multirow{2}{*}{$t_{G}$} & \multirow{2}{*}{Width} & \multirow{2}{*}{Depth} & \multicolumn{3}{c|}{Invalid Line Count} \\ \cline{6-8}
    & & & & & 1 & 3 & 5 \\ \hline \hline
    \multirow{2}{*}{Triangle} & \multirow{2}{*}{\texttt{AerSimulator}} & 0 & 11 & 2 & 1,515 & 533 & -\\
    & & 1 & 11 & 33 & - & 2,048 & -\\\hline \hline
    \multirow{2}{*}{Grid} & \multirow{2}{*}{\texttt{AerSimulator}} & 0 & 22 & 2 & 405 & 1,264 & 379 \\
    & & 1 & 22 & 85 & 1,947 & 74 & 27 \\\hline
\end{tabular}
\caption{Results for both triangle and grid geometries on clean simulated backends with 2,048 shots per run, and no circuit transpilation, with $y = 2$, $t_{A} = 1$ in all cases. Counts for number of invalid lines given on right, e.g. the grid geometry recorded 1,947 assignments with 1 invalid line, with $t_{G}=1$ oracle call. Control runs with no Grover section given in $t_{G}=0$ rows indicative of invalid line distribution across all HV assignments. Results indicate a clear preference for $y^{(1)} \leq y = 2$ for the grid, and $y^{(1)} = 3$ for the triangle, which is expected for a minority of states being unmarked.}
\label{tab:clean_sim_results}
\end{table}

\subsection{Results on Real Quantum Backends}
The algorithm was run on the \texttt{ibm\_marrakesh} backend for both geometries. Results are shown tabularly in Table \ref{tab:real_qc_results} and visually in Appendix \ref{app:Grover_alg_results} for both the real backend and on simulators.  In both \texttt{ibm\_marrakesh} and on noisy fake backend simulators, the circuit was transpiled using the native transpilation process. IBM allows for a variety of multi-control NOT gate implementations using different numbers of ancilla qubits. These different implementations (`noancilla', `recursion', `v-chain') were tested under the `Method' heading.

In all tests as in the clean simulations, 2,048 shots were run with initial threshold $y=2$ was chosen, $t_{A}=1$ and measurement done on $V$. Effects of transpilation on depth are clearly seen in all cases, indicating a noisier output.

\begin{table}[h!]
\centering
\begin{tabular}{|c|cccc|ccc|}
    \hline
    \multirow{2}{*}{Geometry} & \multirow{2}{*}{Backend} & \multirow{2}{*}{Method} & \multirow{2}{*}{Width} & \multirow{2}{*}{Depth} & \multicolumn{3}{c|}{Invalid Line Count} \\ \cline{6-8}
    & & & & & 1 & 3 & 5 \\ \hline \hline
    \multirow{3}{*}{Triangle} & \multirow{3}{*}{\texttt{ibm\_marrakesh}} & noancilla & 11 & 591 & 1,410 & 638 & -\\
    &  & recursion & 11 & 690 & 1,347 & 701 & -\\ 
    & & v-chain & 13 & 575 & 1,412 & 636 & -\\ \hline \hline
    \multirow{3}{*}{Grid} & \multirow{3}{*}{\texttt{ibm\_marrakesh}} & noancilla & 22 & 9,642 & 380 & 1,288 & 380 \\
    &  & recursion & 23 & 8,046 & 383 & 1,290 & 375 \\ 
    & & v-chain & 122 & 5,669 & 391 & 1,288 & 370 \\ \hline
\end{tabular}
\caption{Results for both triangle and grid geometries on \texttt{ibm\_marrakesh} with different ancilla qubit methods. In all cases, $y = 2$, $t_{A} = t_{G} = 1$. Compare invalid line counts with idealised $t_{G}=1$ results from Table \ref{tab:clean_sim_results}. Width and depth requirements change per ancilla methods, with transpilation increasing depth and therefore noise. Invalid line counts do not appreciably differ from control simulations with $t_{G}=0$ (see Figs. \ref{fig:real_qc_results_triangle}, \ref{fig:fake_brisbane_results_triangle} and \ref{fig:real_qc_results_grid} in Appendix \ref{app:Grover_alg_results}). }
\label{tab:real_qc_results}
\end{table}

\section{An Alternate Approach - Quasi-Grover Search}\label{sec:quasi_grover}
The standard Grover algorithm for computing the degree $d$ is limited by high circuit width and depth, making it impractical for geometries larger than the grid on current noisy intermediate-scale quantum (NISQ) computers. This section introduces an alternative``quasi-Grover" approach that eliminates the $X$ and $Y$ registers, significantly reducing CNOT gate requirements. It also condenses the $L$ register to a single qubit and removes the need for an initial threshold guess $y$, instead directly increasing measurement probabilities of states with extremal numbers of invalid lines. While this offers substantial circuit simplification, it comes at the cost of reduced maximal measurement probability for target states in larger geometries, an issue addressed in Section \ref{sec:quasi_grover_alts}.

The quasi-Grover method marks elements based on their relative phase, encoding invalid line information there. Unlike standard Grover search, which uses a $\pi$ phase shift (i.e., $e^{i\pi}=-1$), this approach employs a generalised phase kickback $e^{i\beta}$ ($0 \leq \beta \leq 2\pi$) using a controlled phase gate $\text{C}_{\beta}$.
For a single qubit $L$ register initialised to $\ket{0}$, the `P' and `N' gates (Fig. \ref{fig:circuit_P_N_def}) are used to flip its state if a line is invalid. This $L$ qubit then controls a $\text{C}_{\beta}$ gate targeting a constant $\ket{1}$ qubit, inducing an $e^{i\beta}$ phase for each invalid line. The $L$ qubit is then reset. These operations are combined into `P$_{\beta}$' and `N$_{\beta}$' gates (Fig. \ref{fig:P_N_gates_beta}). This results in a phase $e^{i\ell_j\beta}$ for each state $\ket{j}$ in the $V$ register, where $\ell_j$ is the number of invalid lines for assignment $j$.

\begin{figure}[h!]
    \centering
    \includegraphics[width=0.4\textwidth]{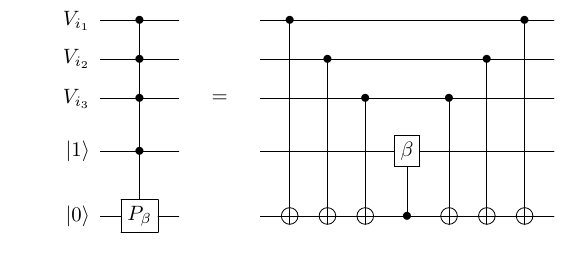}
    \includegraphics[width=0.45\textwidth]{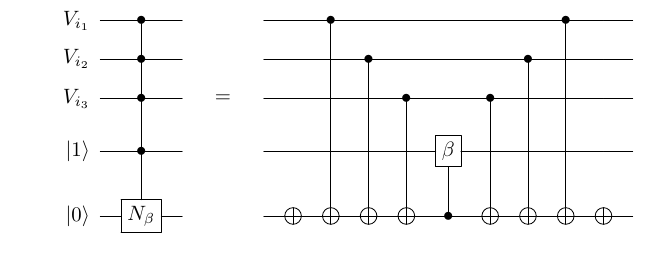}
    \caption{$P_{\beta}$ and $N_{\beta}$ gates defining evolution in \eqref{eq:beta_evolution}. $V$ and $L = \ket{0}$ registers are entangled in the sense that the latter flips to $\ket{1}$ iff the assignment $\ket{V_{i_{1}}, V_{i_{2}}, V_{i_{3}}}$ gives an invalid line, via the first gates up to the control-phase gate $C_{\beta}$. A phase of $\beta$ is then induced if this line is invalid. The entanglement is then undone to restore the $L$ qubit to its original state, allowing it to be reused for all lines in the geometry (see Fig. \ref{fig:quasi_grover_marking}). }
    \label{fig:P_N_gates_beta}
\end{figure}

\begin{figure}[h!]
    \centering
    \includegraphics[width=0.6\textwidth]{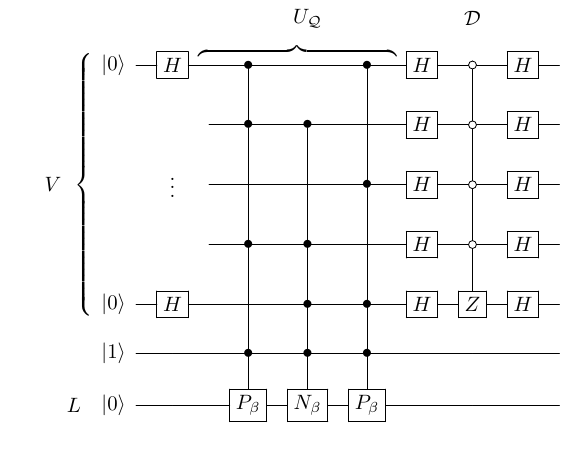}
    \caption{An implementation of a single quasi-Grover query $(\mathcal{D}\otimes I^{\otimes 2})U_{\mathcal{Q}}$ with generalised phase kickbacks composing $U_{\mathcal{Q}}$ - see Fig. \ref{fig:P_N_gates_beta}. The usual diffusion operator $\mathcal{D}$ acts on $V$ and amplifies the amplitudes of those states with the extremal phases. The advantage here is only one qubit needed for the $L$ register, no $X$ or $Y$ registers at all, and no need to specify a threshold $y$.}
    \label{fig:quasi_grover_marking}
\end{figure}

Again we initialise with a superposition over $V$ as before, but now ``mark" states via relative phase shift that count the number of invalid lines. This is done via $U_{\mathcal{Q}}$, composed of `P$_{\beta}$' and `N$_{\beta}$' gates for all lines, see \eqref{eq:beta_evolution} and Fig. \ref{fig:quasi_grover_marking}.

\begin{equation}\label{eq:beta_evolution}
\begin{split}
\ket{\psi_0} &= \frac{1}{\sqrt{n}}\sum_j \ket{j}\ket{1}\ket{0} \\
\mapsto U_Q \ket{\psi_0} &= \frac{1}{\sqrt{n}}\sum_j e^{i\ell_j\beta}\ket{j}\ket{1}\ket{0}
\end{split}
\end{equation}

Each quasi-Grover query follows this unitary operation with the standard diffusion operator $\mathcal{D}$ acting only on the $V$ register. Recall $\mathcal{D}\alpha_j\ket{j} = (2\overline{\alpha} - \alpha_j)\ket{j}$, with $\overline{\alpha}$ being the mean amplitude. To ensure maximal effect of $\mathcal{D}$ on our states with minimal or maximal invalid lines, we need $U_{\mathcal{Q}}$ to position their amplitudes farthest from the mean. To achieve this we set $\beta = \frac{2\pi}{L}$, defining $\delta_j := \ell_j \frac{2\pi}{L}$ as the relative phase of a given state $\ket{j}$ after marking, with $0 \leq \delta_j \leq 2\pi$. This employs an advantageous symmetry in our 3-point line geometries, bolstered by the following observation:

\begin{obs}\label{obs:half_invalid_lines}
For a geometry $\mathcal{G}$ with $V$ vertices, $L$ lines, and complete set of HV assignments, the distribution of invalid line counts across assignments resembles a normal distribution, with roughly half of all assignments having $\ell \approx \frac{L}{2}$ invalid lines.
\end{obs}

For a contextual geometry the distribution is not exactly normal, as there are no assignments with $\ell < d$ or $\ell > L-d$ invalid lines, by definition of $d$. While the precise statement of Observation \ref{obs:half_invalid_lines} will not be proven here, the bilateral symmetry of the invalid line distributions will be:

\begin{proposition}
For a geometry $\mathcal{G}$ with odd number of points per line, the distribution of invalid lines across all HV assignments is symmetric under interchange of $\ell \leftrightarrow L-\ell$.
\end{proposition}
\begin{proof}
Examining any given HV assignment $j$ and any given odd-point line $L_{i}$, assigning the bit-flipped values $\tilde{j}$ renders $L_{i}$ invalid (resp. valid) if $j$ rendered it valid (resp. invalid). Thus the interchange $j \leftrightarrow \tilde{j}$ exchanges valid and invalid lines $\ell \leftrightarrow L - \ell$, and as the set of HV assignments is complete (for each $j$ there exists $\tilde{j}$) we get the desired result. 
\end{proof}

In particular, this means that each assignment $j$ is phase shifted by $\delta_{j}$ via $U_{\mathcal{Q}}$, and its bitflipped assignment $\tilde{j}$ is shifted by $\delta_{\tilde{j}} = (L - \ell_j) \frac{2\pi}{L} \equiv -\delta_{j}$. For $\ell = d$, the phase shift is minimal, and for $\ell = L - d$ it is maximal, see Fig. \ref{fig:phases_plot}. We present exact distributions of invalid lines for four geometries: the grid (Fig. \ref{fig:grid_example}), the two-spread (Fig. \ref{fig:doily}, right), the doily (Fig. \ref{fig:doily}, left) and finally the eloily (Fig. \ref{fig:eloily}). These are given graphically in Fig. \ref{fig:invalid_line_distributions} and numerically in Table \ref{tab:invalid_line_distributions} for the first three, and all satisfy Observation \ref{obs:half_invalid_lines}.

The consequence of this is that after a quasi-Grover oracle $U_{\mathcal{Q}}$, the mean amplitude $\overline{\alpha}$ will be real (owing to the conjugate symmetry $\delta_{\tilde{j}} = -\delta_{j}$) and negative (owing to the heavily-weighted cluster of amplitudes near $\delta_{j} \approx \pi$, via Observation \ref{obs:half_invalid_lines}). The extremal phases associated to $\ell \in \{d, L-d\}$ thus have the largest distance from the mean, as desired.

\begin{algorithm}[h!]
\caption{The quasi-Grover search algorithm for $d$}\label{alg:quasi_grover}
\begin{algorithmic}
\State $\texttt{circuit} \gets \texttt{QuantumCircuit}(V+2, V)$\Comment{Circuit with fewer qubits than Algorithm \ref{alg:whole_alg}}\;
\State \texttt{circuit.h}$(V)$\;
\For{$t \in (1, \dots, t_{\text{opt}})$}
    \For{$i \in (1, \dots, L)$}
        \If{$L_{i}$ positive:}
            \State \texttt{circuit}.$P_{\beta}(V, \ket{1}, \ket{0})$\;
        \Else;
            \State \texttt{circuit}.$N_{\beta}(V, \ket{1}, \ket{0})$\;
        \EndIf;
    \EndFor;
    \State \texttt{circuit}.$\mathcal{D}(V)$\;
\EndFor;
\State $S_{V} \gets $\texttt{circuit.measure}$(V)$\;
\State $x \gets $\texttt{invalid}$(S_{V})$\Comment{classically compute $x$ invalid lines}\;
\State $y \gets \text{min}(x, L-x)$\;
\State \Return $y$
\end{algorithmic}
\end{algorithm}

\begin{figure}[h!]
    \centering
    \includegraphics[width=0.4\textwidth]{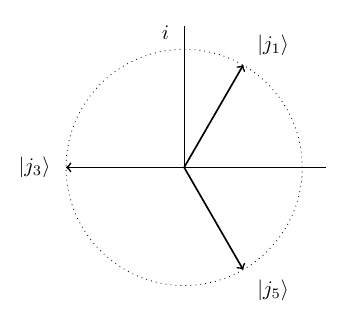}
    \includegraphics[width=0.5\textwidth]{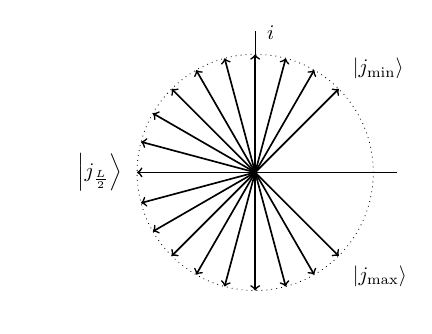}
    \caption{Phases $e^{i\delta_{j}}$ in the complex plane, with dotted circle showing magnitude $\frac{1}{\sqrt{n}}$. Left is for the grid Fig. \ref{fig:grid_example}, with $\ket{j_{1}}$ assignments with 1 invalid line and phase $\delta_{j_{1}} = \frac{\pi}{3}$, etc. Right is general symmetric geometry with $\ket{j_{\text{min}}}, \ket{j_{\text{max}}}$ the assignments with the minimal / maximal number of invalid lines. Under computational checks, most assignments are clustered around $\ket{j_{\frac{L}{2}}}$ with roughly half of all lines invalid.}
    \label{fig:phases_plot}
\end{figure}

\begin{figure}[h!]
    \centering
    \includegraphics[width = 0.8\textwidth]{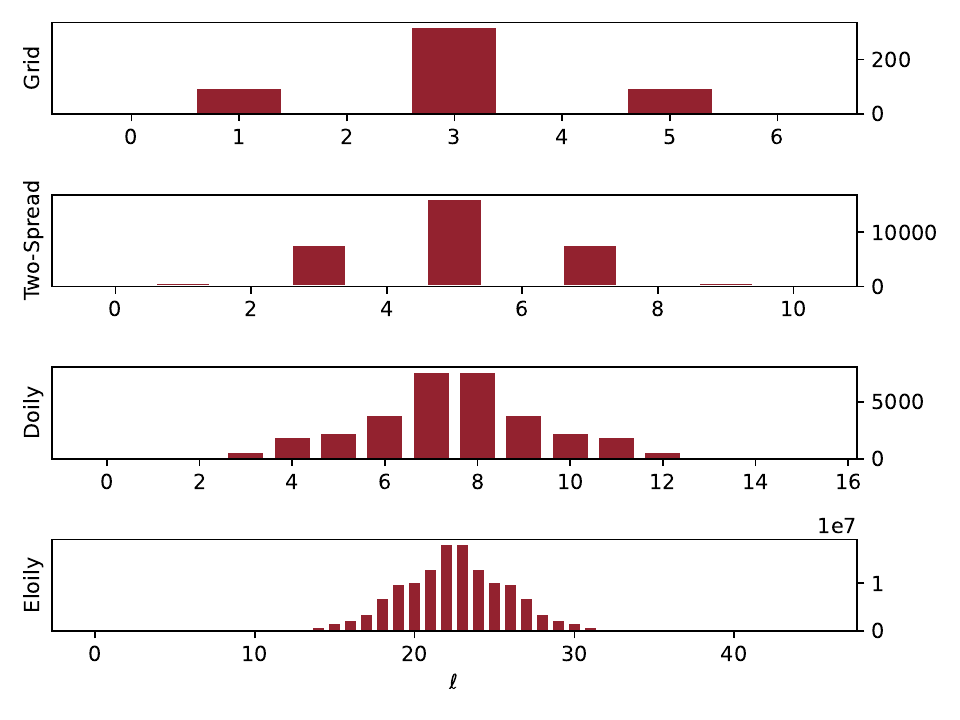}
    \caption{Distributions of assignments $|j_{\ell}|$ for various invalid line values $\ell$ of the indicated geometries. See Table \ref{tab:invalid_line_distributions} for numerical values for the first three. Geometries with odd number of vertices per line have clear symmetry about $\ell \leftrightarrow L-\ell$. For the eloily, nonzero $|j_{\ell}|$ begins at $\ell=9$, with $|j_{9}|=2,560$ not visible due to insignificant size relative to $n=2^{27}$.}
    \label{fig:invalid_line_distributions}
\end{figure}

\begin{table}[h!]
\centering
\begin{tabular}{|c|ccc|} \hline
$\ell$ & Grid & Two-Spread & Doily \\ \hline \hline
1 & ${3}\cdot{2^{5}}$ & ${5}\cdot{2^{7}}$ & \\
2 & & & \\
3 & ${5}\cdot{2^{6}}$ & ${15}\cdot{2^{9}}$ & ${5}\cdot{2^{7}}$ \\
4 & & & ${15}\cdot{2^{7}}$ \\
5 & ${3}\cdot{2^{5}}$ & ${63}\cdot{2^{8}}$ & ${9}\cdot{2^{8}}$ \\
6 & & & ${15}\cdot{2^{8}}$ \\
7 & & ${15}\cdot{2^{9}}$ & ${15}\cdot{2^{9}}$ \\
8 & & & ${15}\cdot{2^{9}}$ \\
9 & & ${5}\cdot{2^{7}}$ & ${15}\cdot{2^{8}}$ \\
10 & & & ${9}\cdot{2^{8}}$ \\
11 & & & ${15}\cdot{2^{7}}$ \\
12 & & & ${5}\cdot{2^{7}}$ \\ \hline \hline
$n$ & $2^{9}$ & $2^{15}$ & $2^{15}$\\ \hline
\end{tabular}
\caption{Number of assignments $|j_{\ell}|$ for a given number $\ell$ of invalid lines, for symmetric geometries grid, two-spread and doily. In all cases, $\sum_{\ell}|j_{\ell}| = n$. As shown here $\ell \approx \frac{L}{2}$ contains about half of all assignments.}
\label{tab:invalid_line_distributions}
\end{table}

Now that we have numerical justification that a quasi-Grover query $\mathcal{D} U _{\mathcal{Q}}$ should work, we can examine its effectiveness. Simulating Algorithm \ref{alg:quasi_grover} on the four mentioned geometries, we can compute after each query the probability of measuring a state $j$ with invalid lines $\ell$, denoted $P(\ell)$. This is plotted over $t$ queries for $d \leq \ell \leq \frac{L}{2}$ in Fig. \ref{fig:P_ell_evolution}, and the value of $P(d)$ at the optimal query time $t_{\text{opt}}$ is given in Table \ref{tab:P_d_numeric_vals}. The value of $P(\ell)$ is periodic in $t$, so $t_{\text{opt}}$ is defined as the time coordinate of the first local maxima. The effectiveness (or lack thereof) of this quasi-Grover approach can be seen in Figs. \ref{fig:P_l_distributions}, \ref{fig:P_l_distributions_eloily}, where the highlighted data indicate probabilities of measuring a desired state compared to baseline (no queries). The algorithm is very effective for small geometries like the grid, but loses accuracy as the size of the geometry grows. Despite this decreasing accuracy, the advantages remain that this approach uses far fewer qubits than the standard Grover approach, and does not require an artificial threshold guess $y$. In order to address the accuracy issues, we introduce some tweaks to the quasi-Grover method in the next sections.

\begin{table}[h!]
\centering
\begin{tabular}{|cc|cc|} \hline
Geometry & $t_{\text{opt}}$ & $P(d)|_{t=0}$ & $P(d)|_{t_{\text{opt}}}$ \\ \hline
Grid & 2 &  0.1870 & 0.4999 \\
Two-Spread & 4 & 0.0195 & 0.2859 \\
Doily & 2 & 0.0195 & 0.0997 \\ 
Eloily & 2 & 0.00002 & 0.000155 \\ \hline
\end{tabular}
\caption{Increase in $P(d)$ - probability of measuring a state with $d$ invalid lines - from no quasi-Grover queries ($t = 0$) to value after optimal queries $t_{\text{opt}}$, shown in Fig. \ref{fig:P_ell_evolution}. Probability of extracting $d$ from measurement is given by $2P(d)$ as $P(L-d)=P(d)$ by symmetry.}
\label{tab:P_d_numeric_vals}
\end{table}

\begin{figure}[h!]
\centering
\includegraphics[width=0.45\textwidth]{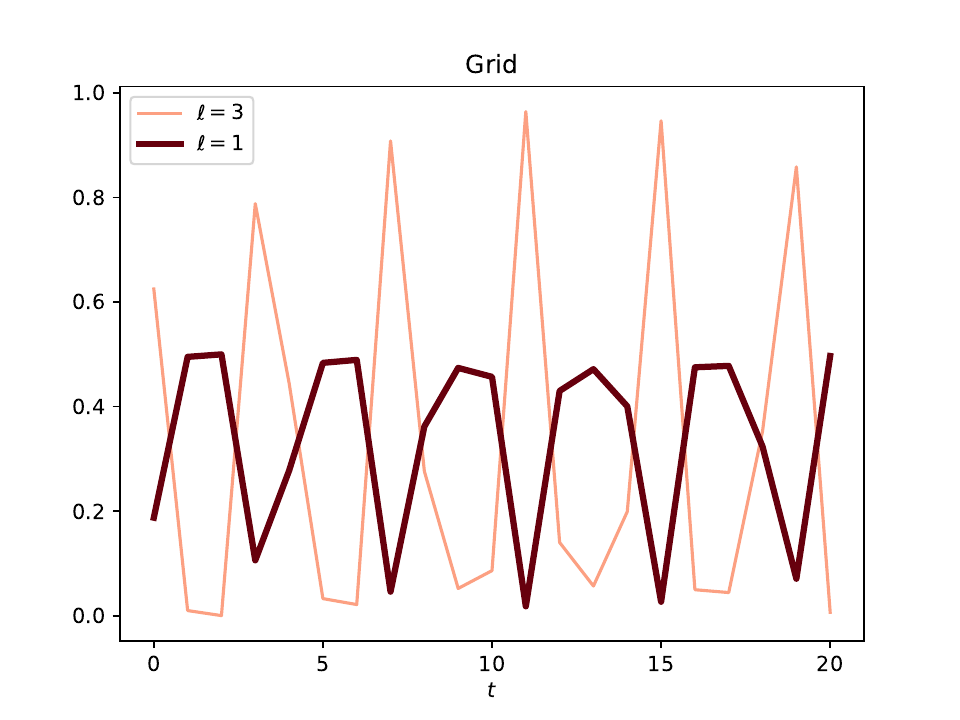}
\includegraphics[width=0.45\textwidth]{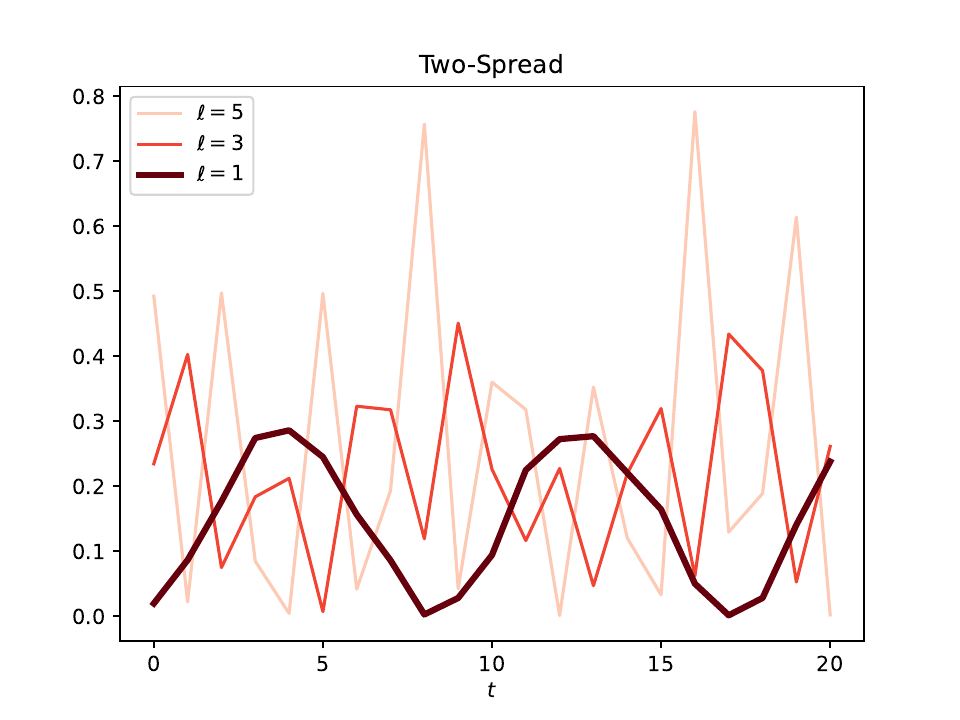}
\includegraphics[width=0.45\textwidth]{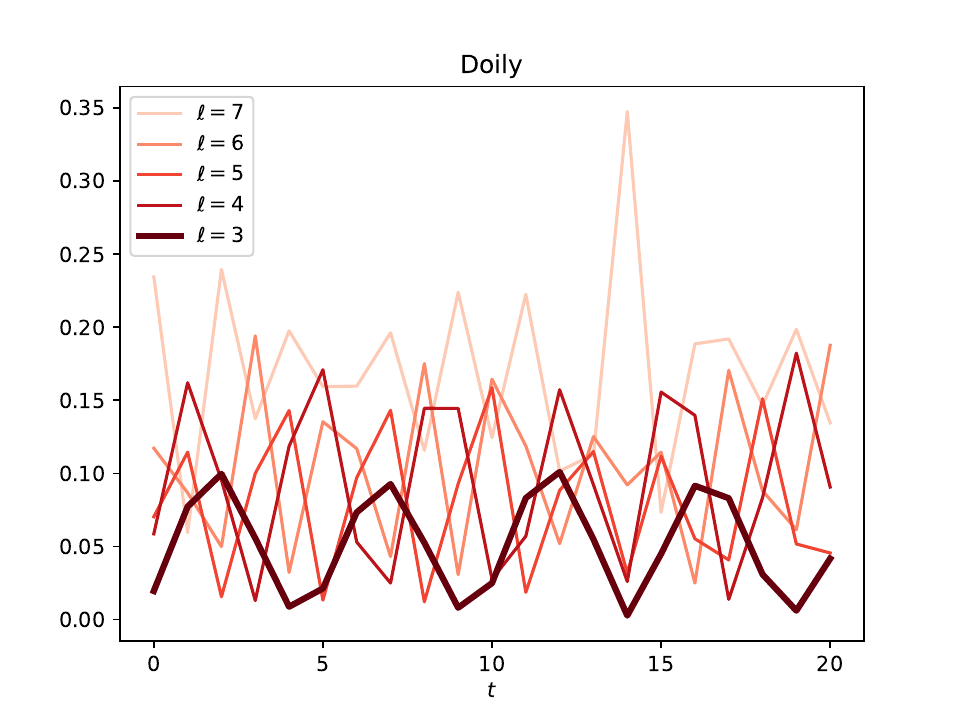}
\includegraphics[width=0.45\textwidth]{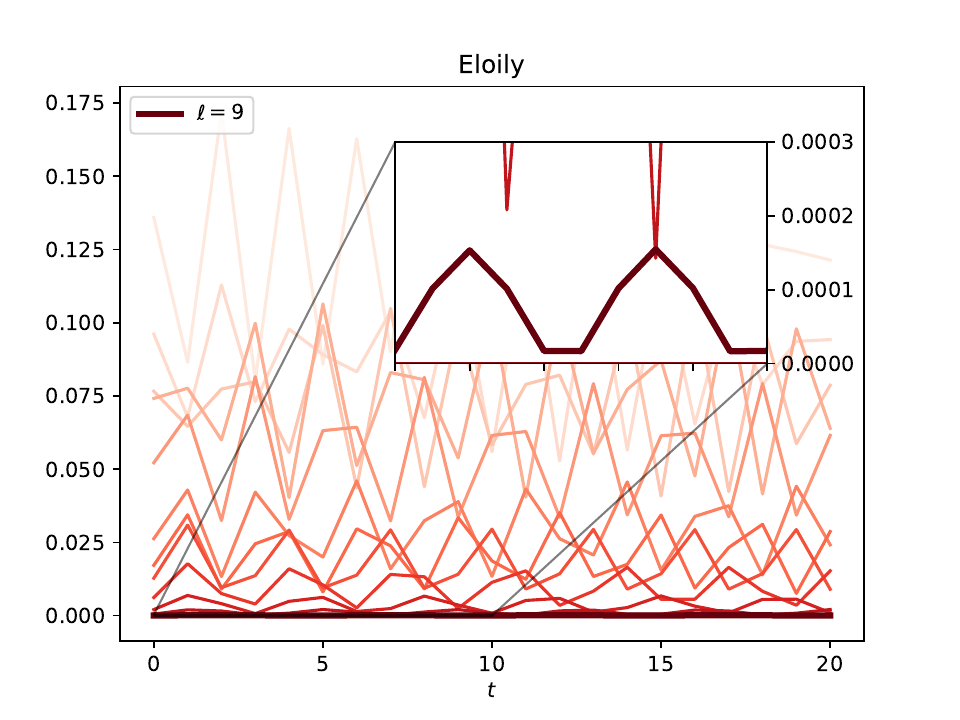}
\caption{Evolution of $P(\ell)$ over $t$ quasi-Grover queries for grid (top left), two-spread (top right), doily (bottom left) and eloily (bottom right, with insert on $t\in [0, 10]$). Different values of $\ell \leq \frac{L}{2}$ shown, with larger $\ell$ equal to those of $L - \ell$ via symmetry. Values for $P(d)$ given in bold, increasing from baseline to optimal $t$ with numeric values in Table \ref{tab:P_d_numeric_vals}.}
\label{fig:P_ell_evolution}
\end{figure}

\begin{figure}[h!]
\centering
\includegraphics[width=0.8\textwidth]{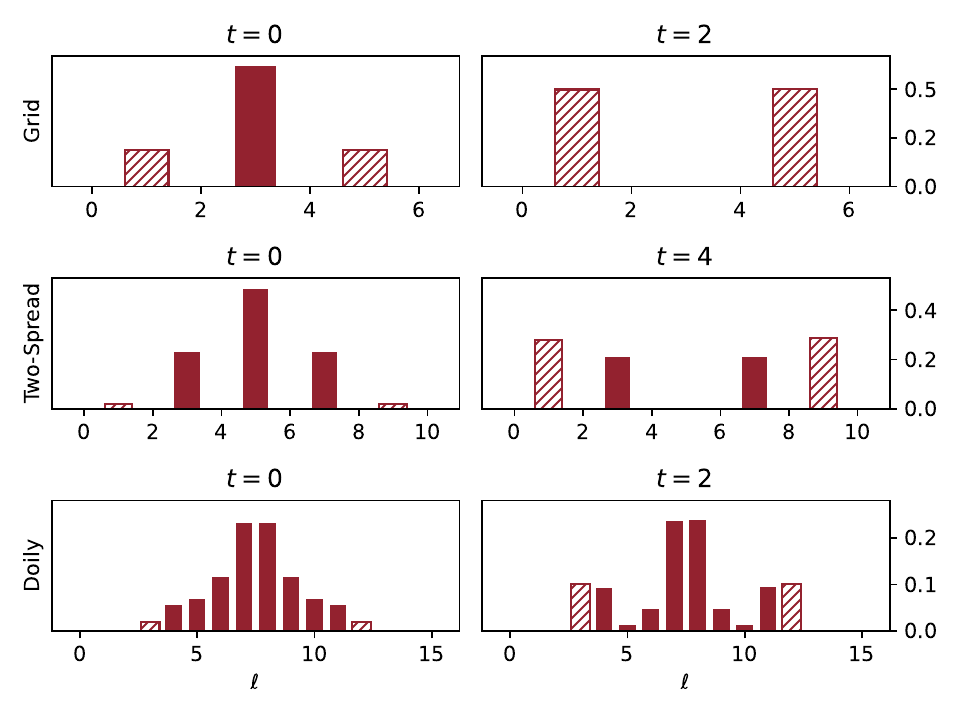}
\caption{Probabilities $P(\ell)$ for labelled geometries both before (left) and after (right) optimal quasi-Grover queries. Distributions on left are normalised versions of Fig. \ref{fig:invalid_line_distributions}, on right are measurement results after simulated circuits without noise. Quasi-Grover protocol increases values of $\{P(d), P(L-d)\}$ (highlighted) to specific figures in Table \ref{tab:P_d_numeric_vals}. Success of retrieving $d$ from measurement falls as geometry gets larger -- maximal values for grid is roughly $0.9998$, and for the other geometries given in Table \ref{tab:P_d_numeric_vals}.}
\label{fig:P_l_distributions}
\end{figure}

\begin{figure}[h!]
\centering
\includegraphics[width=0.8\textwidth]{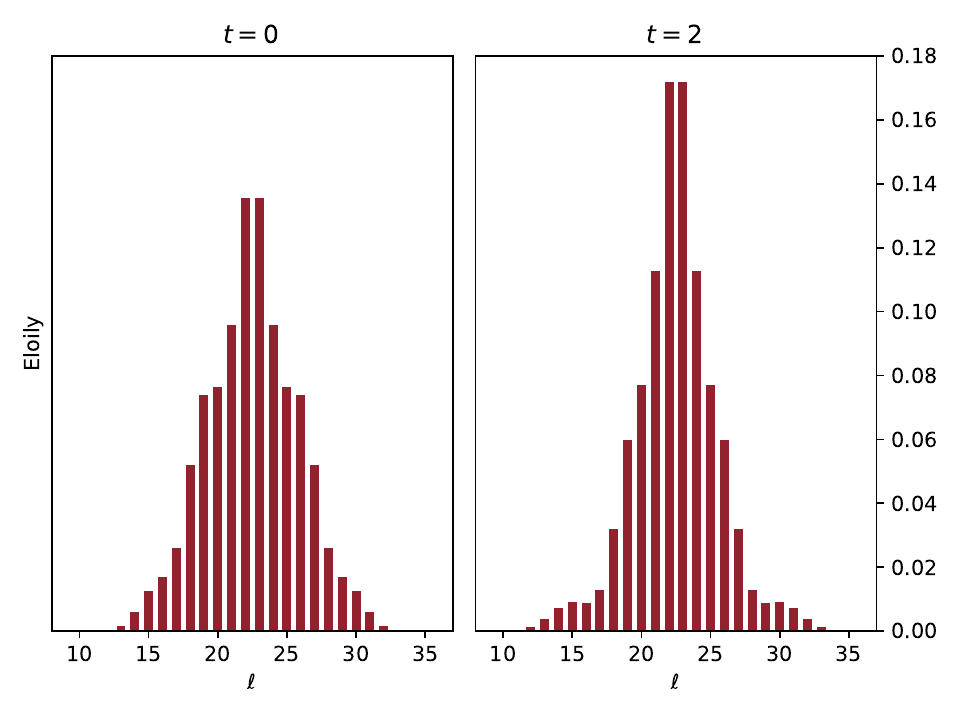}
\caption{Probabilities $P(\ell)$ for the eloily as in Fig. \ref{fig:P_l_distributions}. Change in $P(d=9)$ imperceptible, but value rises from $0.00002$ to $0.000155$.}
\label{fig:P_l_distributions_eloily}
\end{figure}

\subsection{Dynamic $\beta$ Values}\label{sec:quasi_grover_alts}
Despite its advantages, the quasi-Grover method (with $U_Q$) suffers from a low success rate in retrieving $d$, especially for larger geometries like the doily ($P(d) = 0.1994$) and eloily ($P(d)$ negligible). This can be addressed by dynamically adjusting the kickback angle. Returning to $U_{\mathcal{Q}}$ for modelling, we replace the fixed $\beta$ with a query-dependent multiple $b_t\beta$ ($b_t \in \mathbb{N}$, $0 \leq b_t < L$) at each query $t$. The value of $b_{t}$ is chosen to maximise the distance 

\begin{equation}\label{eq:distance_D_t}
D_t(\ell) := |\overline{\alpha_t} - e^{ib_t\ell\beta}\alpha_{\ell,t}|
\end{equation}

for $\ell=d$, where after $t$ queries $\overline{\alpha_{t}}$ is the mean amplitude and $\alpha_{\ell,t}$ the amplitude of a state with $\ell$ invalid lines. This boosts the efficacy of $\mathcal{D}$, both in terms of per-query effect on states with extremal $\ell$, and over longer timespans than the usual quasi-Grover method. 

How $b_{t}$ is chosen is by initially simulating the evolution of the state, calculating $D_t(\ell)$ over all values of $b_{t} \in [0, L)$, as each value determines a specific $\overline{\alpha_{t}}$. The optimal value is chosen for each $t$, which converges on $b_{t}=0$ once an upper limit on $P(d)$ is reached. A more analytic approach is not possible without knowledge of the distribution $|j_{\ell}|$.

Figure \ref{fig:P_evolution_beta_multipliers} illustrates the effect of taking $\text{max}_{b_{t}}D_{t}(\ell)$ at each query step $t$, showing dramatic increases in $P(d)$ for various geometries. Table \ref{tab:beta_multipliers_values} provides the optimal $b_t$ sequences, maximal $P(d)$ values, and corresponding optimal query times $t'_{\text{opt}}$. Considering both $P(d)$ and $P(L-d)$ as successes, the success rates are 0.9998 (grid), 0.9342 (two-spread), 0.9029 (doily), and 0.9894 (eloily), significantly higher than with fixed $\beta$.

\begin{figure}[h!]
    \centering
\includegraphics[width=0.45\textwidth]{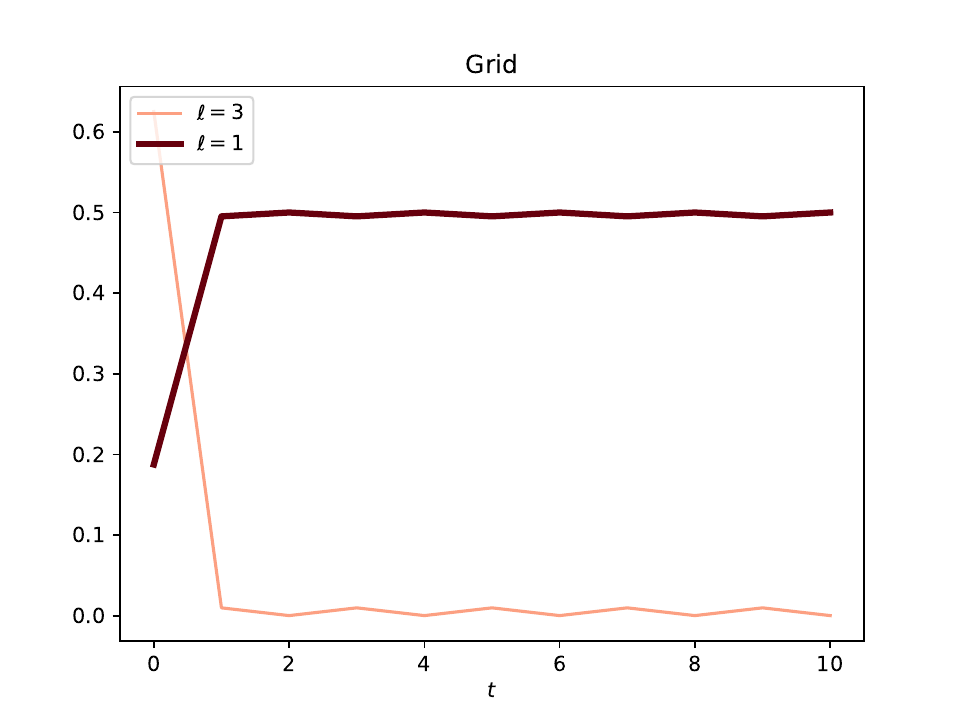}
\includegraphics[width=0.45\textwidth]{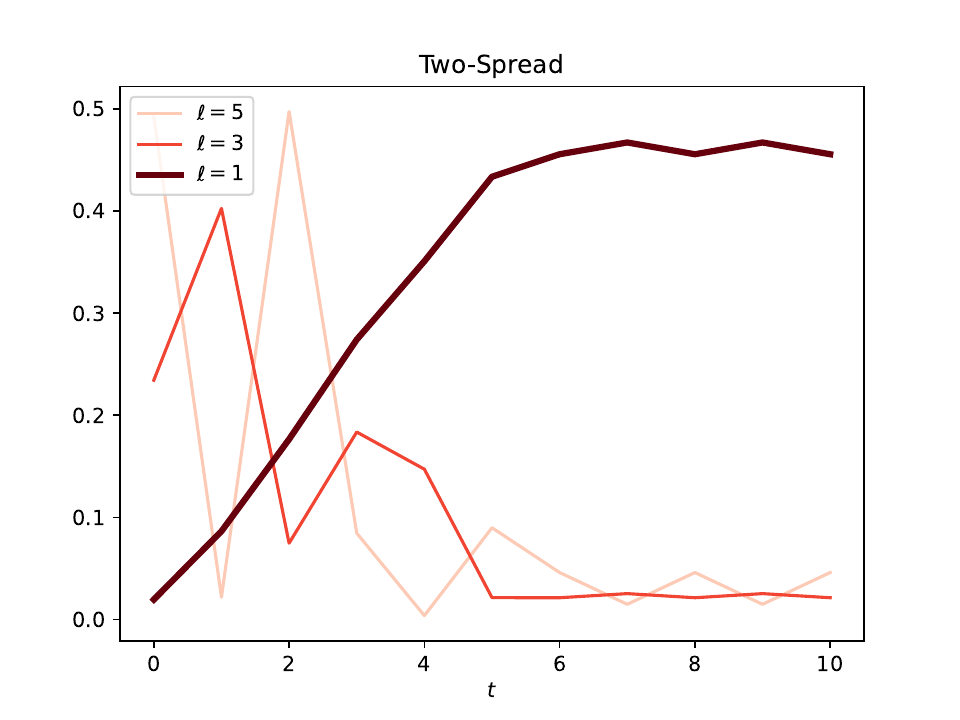}
\includegraphics[width=0.45\textwidth]{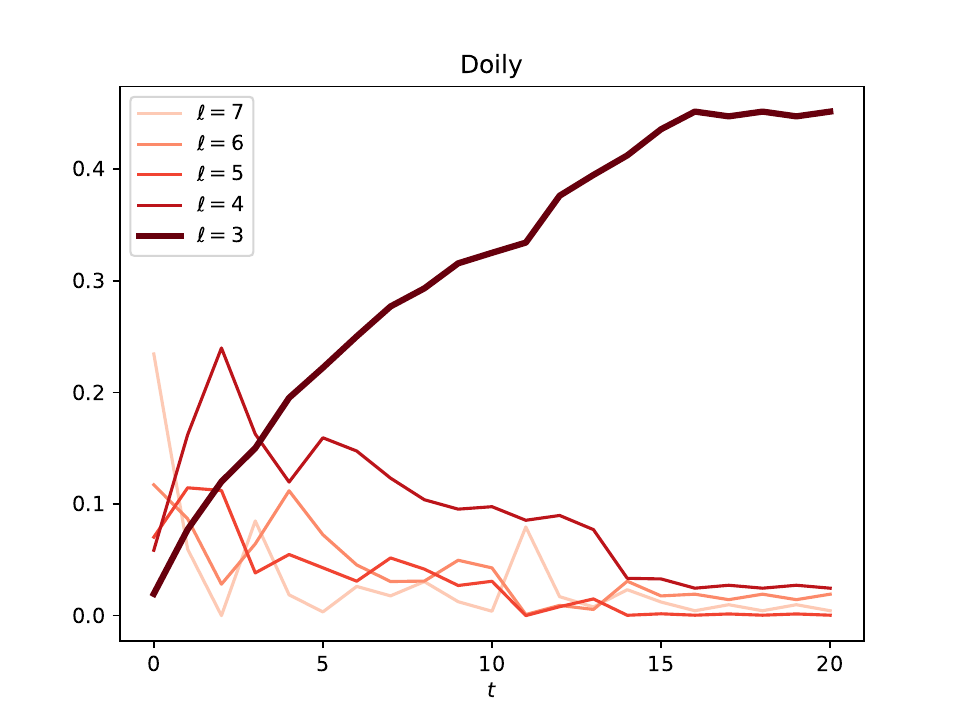}
\includegraphics[width=0.45\textwidth]{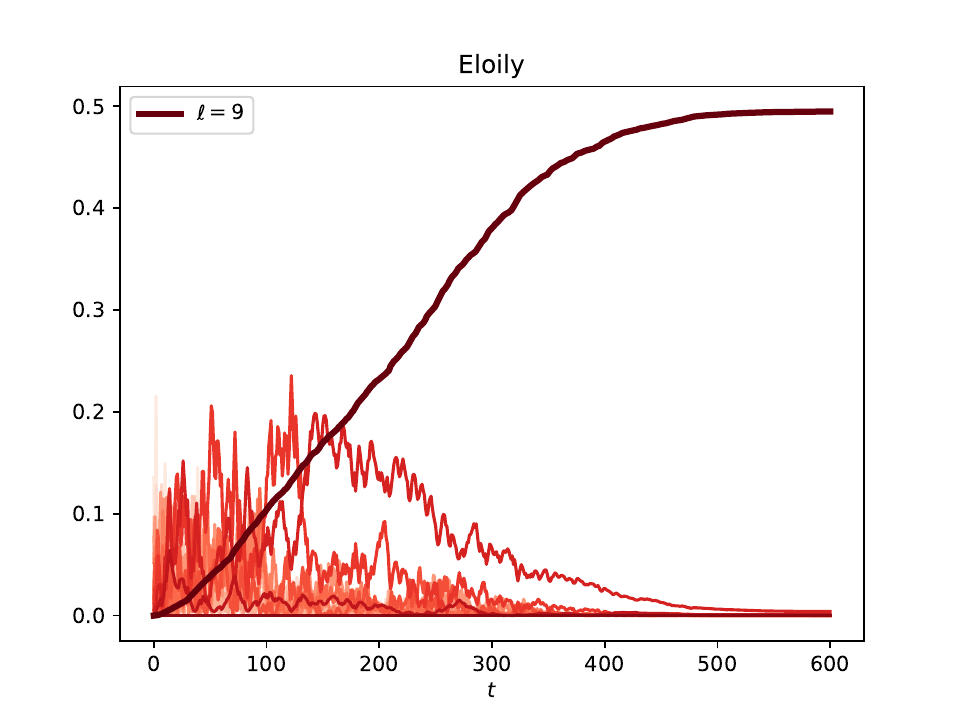}
    \caption{Evolution of $P(\ell)$ as in Fig. \ref{fig:P_ell_evolution} but with nontrivial $\beta$ multipliers $b_{t}$ at each query $t$ (note the higher query counts). Values of $b_{t}$, $\text{max}(P(d))$ and $t_{\text{opt}}'$ are given in Table \ref{tab:beta_multipliers_values}.}
    \label{fig:P_evolution_beta_multipliers}
\end{figure}

\begin{table}[h!]
\centering
\begin{tabular}{|c|l|cc|}
    \hline
     Geometry & $b_{t}$ & $t_{\text{opt}}'$ & $\text{max}(P(d))$ \\ \hline
     Grid & $(4, 1, 0)$ & $2$ & $0.49992$ \\
     Two-Spread & $(1, 1, 1, 4, 9, 1, 7, 0)$ & $7$ & $0.46710$ \\
     Doily & $(14, 13, 14, 13, 3, 5, 5, 14, 13, 11, 13, 13, 6, 6, 1, 12, 14, 0)$ & $16$ & $0.45146$ \\ 
     Eloily & & $587$ & $0.49469$ \\ \hline
\end{tabular}
\caption{Optimal $\beta$ multipliers $b_{t}$ for different geometries to give evolution in Fig. \ref{fig:P_evolution_beta_multipliers}. Shown are the maximal values for $P(d)$, the probability of measuring $d$ invalid lines, and the number of queries $t_{\text{opt}}'$ to find $\text{max}(P(d))$. All nonzero $b_{t}$ shown, any values for larger $t$ than displayed give $b_{t}=0$. Full eloily data not shown due to space constraints.}
\label{tab:beta_multipliers_values}
\end{table}

\subsection{Practical Issues With Dynamic $\beta$ Values}\label{sec:practical_issues}
Despite the theoretical improvements from dynamic $\beta$ multipliers, practical implementation on current NISQ computers remains challenging due to noise. While the quasi-Grover algorithm reduces circuit width from $V + L + 2\log_2 L + 1$ (standard Grover) to $V+2$, circuit depth remains high even for small geometries. Table \ref{tab:quasi_grover_depth} shows substantial depths for the grid, two-spread, and doily after transpilation on the \texttt{ibm\_kingston} backend with dynamic $\beta$ multipliers. As shown in Figure \ref{fig:quasi_grover_sim_v_kingston}, the additional noise from these large depths renders the results from real backends inaccurate, masking the theoretical signal.

\begin{table}[h!]
\centering
\begin{tabular}{|c|c|} \hline
Geometry & Depth \\ \hline
Grid & 3,960 \\
Two-Spread & 56,628 \\
Doily & 132,693 \\ \hline
\end{tabular}
\caption{Circuit depth on \texttt{ibm\_kingston} after transpilation for the quasi-Grover method, for beta multipliers with optimal queries. The large depth for geometries even as small as the grid render the noise levels too high on measurement.}
\label{tab:quasi_grover_depth}
\end{table}

\begin{figure}[h!]
\centering
\includegraphics[width=0.9\textwidth]{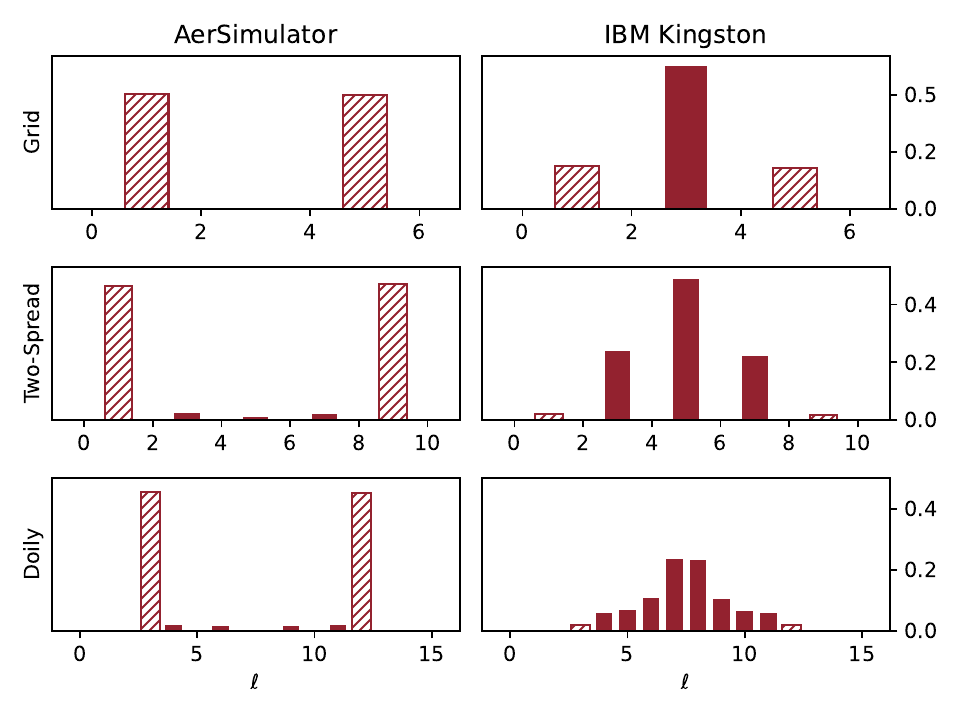}
\caption{Measurement probabilities $P(\ell)$ for the quasi-Grover algorithm on a circuit with optimal beta multipliers and queries (Table \ref{tab:beta_multipliers_values}) for the grid, two-spread and doily. Outputs on \texttt{ibm\_kingston} after 10,000 shots (right) are compared with clean simulator (left), showing that the additional noise in the real backends renders the results inaccurate, with $\ell\in\{d, L-d\}$ probabilities highlighted.}
\label{fig:quasi_grover_sim_v_kingston}
\end{figure}

A theoretical challenge is computing optimal $b_t$ values without prior knowledge of $\ell_j$ distributions or $d$. The former can be approximated by the binomial distribution (for large symmetric geometries considered here),

\begin{equation}
|j_{\ell}|_{\text{binom}} = {L \choose \ell}2^{V - L}
\end{equation}

The latter issue is the more concerning one -- how can we use $D_{t}(d)$ to compute $d$ when it is required in the definition? The answer is to iterate $D_{t}(\ell')$ over multiple values of $\ell'$, forming $b_{t}$ values each time and using those to measure on $V$ until an assignment is found furnishing invalid lines $\ell'$, which then acts as an upper bound on $d$. Each choice of $\ell'$ and subsequent measurement can act to bisect the existing solution space, needing $O(\log{L})$ attempts. Each $b_{t}$ determination per query requires $O(L)$ computations to determine the maximal $D_{t}(\ell')$. Overall, assuming that $t_{\text{opt}}' \approx n^{\frac{1}{3}}$ as it appears here in the computed examples, one has an algorithm for finding $d$ with complexity $O(n^{\frac{1}{3}}L\log{L}) = O(n^{\frac{1}{3}}(\log{n})^{2}\log{\log{n}})$, comparable the standard Grover complexity of $O(n^{\frac{1}{2}})$. Issues with this approach are 1) it is not proven that $t_{\text{opt}}' \approx n^{\frac{1}{3}}$ as seen in the example geometries here, and 2) computing $b_{t}$ on $|j_{\ell}|_{\text{binom}}$ (and then running these $b_{t}$ values on $|j_{\ell}|$ to compute $d$) has maximum success rate of $2\cdot P(d) \approx 0.788$ for the geometries seen here (see Table \ref{tab:results_binom_training}). This reduces the effectiveness of the algorithm.

\begin{table}[h!]
\centering
\begin{tabular}{|c|cc|cc|} \hline
\multirow{2}{*}{Geometry} & \multicolumn{2}{c|}{$|j_{\ell}|$} &  \multicolumn{2}{c|}{$|j_{\ell}|_{\text{binom}}$} \\ \cline{2-5}
& $t_{\text{opt}}'$ & $2\cdot\text{max}(P(d))$ & $t_{\text{opt}}'$ & $2\cdot\text{max}(P(d))$ \\ \hline \hline
Grid & 2 & 0.9998 & 4 & 0.9998 \\ 
Two-Spread & 7 & 0.9342 & 6 & 0.8704 \\ 
Doily & 16 & 0.9029 & 16 & 0.7888 \\ 
Eloily & 587 & 0.9894 & 385 & 0.7872 \\ \hline
\end{tabular}
\caption{Maximum probability $P(d)$ and optimal query time quasi-Grover with $\beta$ multipliers $b_{t}$, computed from real distributions (left) and from approximate (binomial) distributions (right). If the drop in accuracy is not dependent on geometry size then one can expect a constant increase in algorithmic iterations to compensate.}
\label{tab:results_binom_training}
\end{table}

Despite these challenges, the dynamic $\beta$ multiplier approach offers an advantage for problems where $d$ and the approximate invalid line distribution are known, and the goal is to find a maximally-satisfactory assignment. In this scenario, iterating over $\ell'$ is not necessary and the algorithm can provide a desired assignment in $O(n^{\frac{1}{3}}(\log{n})^{2})$ queries. In the more general case where $d$ is not known (i.e., a set of $L$ binary linear constraints over $V$ binary parameters with an odd number of parameters per constraint), this method provides an $O(n^{1/3}(\log{n})^{2}\log{\log{n}})$ solution, still significantly faster than the standard $O(n)$ brute-force method.
\newpage
\section{Conclusion}\label{sec:conclusion_chpt_6}
A quantum algorithm has been developed to find the maximal number of satisfied constraints for \texttt{Max Lin 2} problems, offering a quadratic speedup over classical methods. The standard Grover-based implementation requires $V+L+2X+1$ qubits (where $n=2^V$ inputs, $L$ constraints, $X=\lceil\log_2 L\rceil$) and has a time complexity of $O(\sqrt{n}\log{\log{n}})$, assuming $O(\log{L}) = O(\log{\log{n}})$ measurements. Simulations on a small grid geometry demonstrate that a single run with an initial threshold $y=2$ yields a 95\% probability of measuring the minimal $d=1$ constraint count in an ideal simulation. However, running on real quantum backends like \texttt{ibm\_marrakesh} shows no significant improvement over control algorithms due to noise from transpilation-induced depth increases.

An alternative quasi-Grover algorithm is proposed, which dramatically reduces circuit width and eliminates the need for threshold guessing and iterations. It encodes the invalid line information in the relative phases of the basis states, via a generalisation of the Grover marking oracle. This comes at the cost of reduced accuracy. Two improvements were tested: an alternative unitary operator $U_R$, and dynamically changing $\beta$ phase shifts. The former exploits the binomial-like distribution of invalid line counts, improving the effectiveness of the diffusion operator in the process. The latter maximises the amplitude distance between states of interest and the average at each query step, again improving the performance of the diffusion operator. This provides a search algorithm with $O(n^{1/3}(\log{n})^{2}\log{\log{n}})$ time complexity for finding a maximally-satisfactory assignment. This is particularly useful when approximate values for the degree $d$ and invalid line distribution $|j_\ell|$ are known, but an exact solution is not. 

For the grid, two-spread, doily, and eloily geometries, this method yields probabilities of 0.9998, 0.9342, 0.9029, and 0.9894, respectively, of measuring $d$ after just 2, 7, 16, and 587 queries. These results are comparable to the standard Grover complexity of $O(n^{1/2})$ but with a reduced circuit width of $V+2$. The generalisation of Grover's algorithm to mark states on a scale by encoding numeric information in relative phases is a novel contribution and warrants further investigation. In particular, alternatives to the standard diffusion operator should be considered in conjunction with the phase-encoding quasi-Grover operator $U_{\mathcal{Q}}$, in the search for some action that distinguishes basis states with extremal relative phase.

\section*{Acknowledgments}
This work is supported by the Graduate school EIPHI (contract ANR-17-EURE- 0002) through the project TACTICQ, the Ministry of Culture and Innovation. We acknowledge the use of the IBM Quantum Credits for this work. The views expressed are those of the authors and do not reflect the official policy or position of IBM or the IBM Quantum Experience team. The authors declare that they have no conflicts of interest to disclose, and would like to thank the developers of the open-source framework Qiskit. All code is available at \url{https://github.com/quantcert/quantcert.github.io}.

%\newpage
\printbibliography

%\appendix
%\pagestyle{appendix}
\appendix
\section{Results of Algorithm \ref{alg:whole_alg} on Triangle, Grid}\label{app:Grover_alg_results}
Below are measurement result distributions $P(\ell)$ for various $\ell$ and on the triangle (Fig. \ref{fig:triangle_geom}) and grid (Fig. \ref{fig:grid_example}) geometries after implementing Algorithm \ref{alg:whole_alg}. Implementation was done on a variety of clean and noisy simulators as well as on IBM quantum computing backends, as indicated. A variety of IBM native multi-control NOT gate methods were tested, with 2,048 shots in all cases.

\begin{figure}[h!]
\centering
\includegraphics[width=0.65\textwidth]{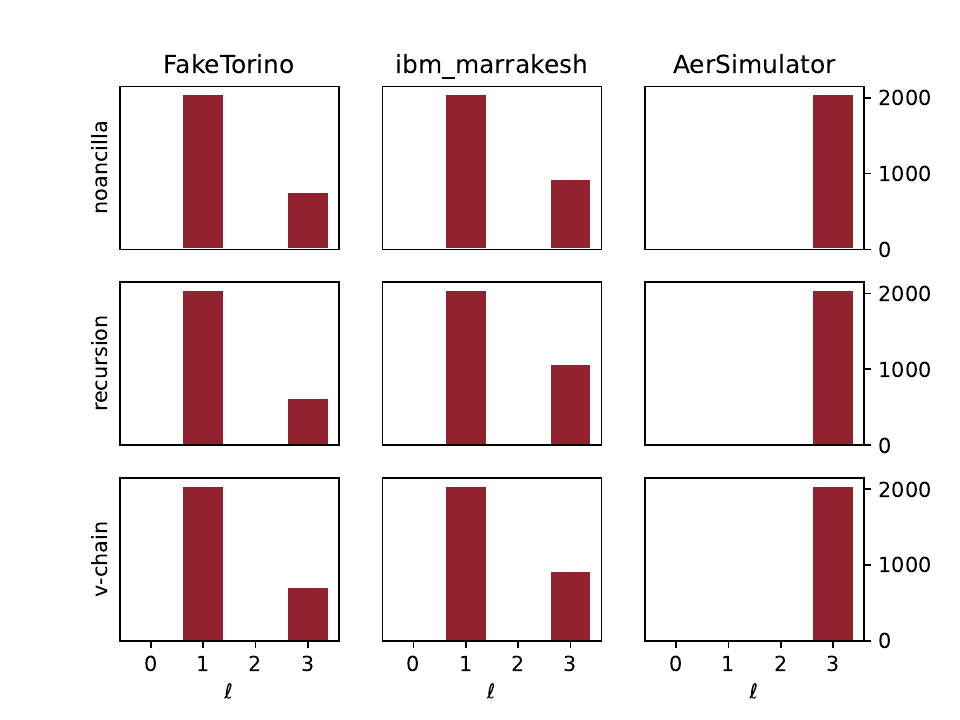}
    \caption{Invalid line distributions for the triangle on different backends with $y=2, t_{A} = 1$ and 2,048 shots. Left to right columns: \texttt{FakeTorino} noisy control simulation with $t_{G}=0$; \texttt{ibm\_marrakesh} results with $t_{G}=1$ (Table \ref{tab:real_qc_results}); \texttt{AerSimulator} clean simulation ($t_{G}=1$). Rows top to bottom: noancilla, recursion, v-chain modes. In noiseless simulation  $y^{(1)}=3$ invalid lines selected as predicted (see Sect. \ref{sec:triangle}). On real backend (middle column), algorithm slightly highlights $y^{(1)}=3$ vs. baseline (left column).}
    \label{fig:real_qc_results_triangle}
\end{figure} 

\begin{figure}[h!]
\centering
\includegraphics[width=0.65\textwidth]{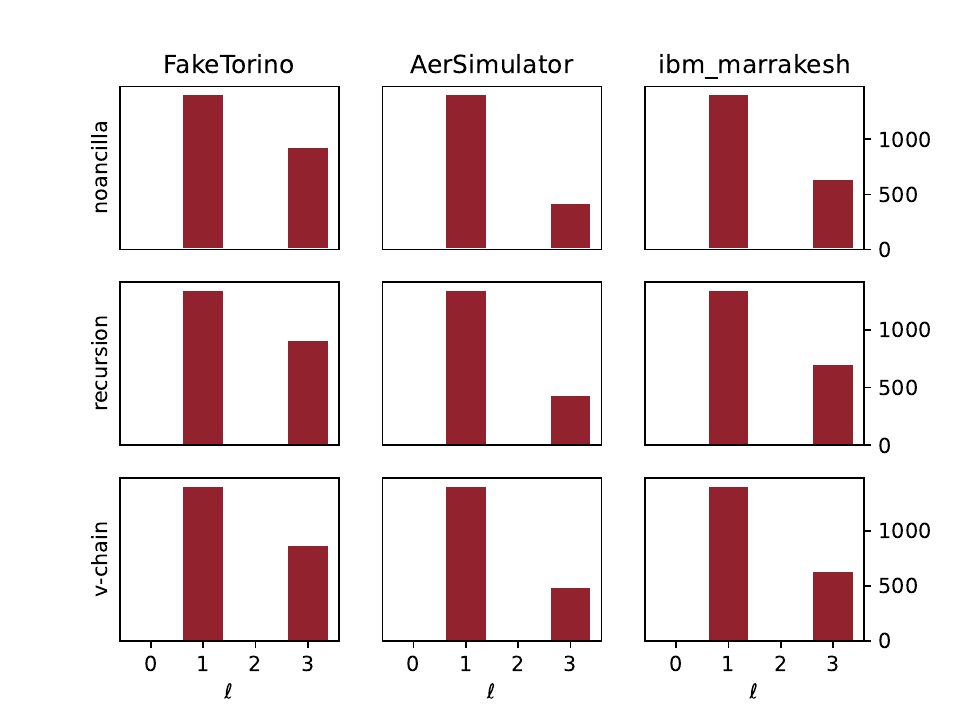}
    \caption{Invalid line distributions for the triangle on \texttt{FakeTorino}, \texttt{AerSimulator} using full noisy Grover simulation, with $y=2, t_{A} = 1$. Left to right columns: \texttt{FakeTorino} simulation with $t_{G}=1$; noisy \texttt{AerSimulator} simulation with $t_{G}=1$; \texttt{ibm\_marrakesh} results with $t_{G}=1$ (Table \ref{tab:real_qc_results}). Rows top to bottom: noancilla, recursion, v-chain modes. Results are not indicative of a preference towards $y^{(1)}=3$, due to transpilation-induced noise in all cases.}
    \label{fig:fake_brisbane_results_triangle}
\end{figure} 

\begin{figure}[h!]
\centering
\includegraphics[width=0.75\textwidth]{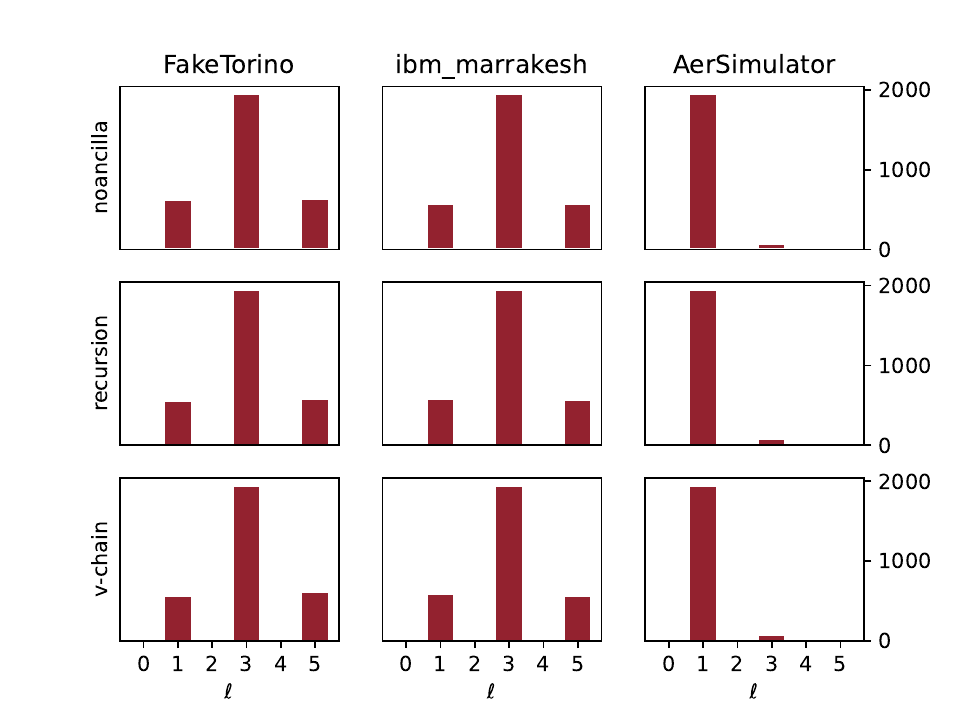}
    \caption{Invalid line distributions for the grid with $y=2, t_{A} = 1$. Left to right columns: \texttt{FakeTorino} control noisy simulation with $t_{G}=0$; \texttt{ibm\_marrakesh} results with $t_{G}=1$ (Table \ref{tab:real_qc_results}); noiseless \texttt{AerSimulator} simulation ($t_{G}=1$). Rows top to bottom: noancilla, recursion, v-chain modes. Simulated noisy circuit too large to run full Grover algorithm ($t_{G}=1$) on \texttt{FakeTorino} or \texttt{AerSimulator}. In noiseless simulation no ancilla qubits needed so results replicated on 3 rows, which shows clear preference toward $y^{(1)}=1$ invalid lines, in agreement with the fact that $d=1$ for the grid. As with triangle example, transpilation-induced noise renders NISQ output (middle column) indistinguishable from control (left).}
    \label{fig:real_qc_results_grid}
\end{figure}

\end{document}